\documentclass[manyauthors]{fundam}

\usepackage{hyperref}
\usepackage{mathrsfs}
\usepackage{amssymb, amsmath, mathrsfs}
\usepackage{graphicx}
\usepackage{color}
\usepackage{framed}
    \usepackage{diagbox}
\input{epsf}

\newcommand{\bd}{\begin{description}}
\newcommand{\ed}{\end{description}}
\newcommand{\bi}{\begin{itemize}}
\newcommand{\ei}{\end{itemize}}
\newcommand{\be}{\begin{enumerate}}
\newcommand{\ee}{\end{enumerate}}
\newcommand{\beq}{\begin{equation}}
\newcommand{\eeq}{\end{equation}}
\newcommand{\beqs}{\begin{eqnarray*}}
\newcommand{\eeqs}{\end{eqnarray*}}

\definecolor{DarkGreen}{rgb}{0.2, 0.6, 0.3}

\newtheorem{problem}{Problem}
\newtheorem{observation}{Observation}[section]
\usepackage{algorithm}
\usepackage{algorithmicx}
\usepackage{algpseudocode}

\begin{document}


\setcounter{page}{141}
\publyear{24}
\papernumber{2176}
\volume{191}
\issue{2}

\finalVersionForARXIV

 \title{Perturbation Results for Distance-edge-monitoring Numbers\thanks{Supported by the National Science Foundation
                                                of China (Nos. 12061059, 12271362), the Qinghai Key Laboratory of Internet of Things
                                                Project (2017-ZJ-Y21), and by the ANR project TEMPOGRAL (ANR-22-CE48-0001).}}

\author{Chenxu Yang \\
School of Computer\\
Qinghai Normal University\\
 Xining, Qinghai 810008, China\\
 cxuyang@aliyun.com
\and
 Ralf Klasing \thanks{Corresponding author: Universit\'{e} de Bordeaux, Bordeaux INP, CNRS, LaBRI, UMR 5800, Talence, France.\newline \newline
                    \vspace*{-6mm}{\scriptsize{Received May 2023; \ accepted December 2023.}}}\\
 Universit\'{e} de Bordeaux\\
 Bordeaux INP, CNRS, LaBRI\\
  UMR 5800, Talence, France\\
  ralf.klasing@labri.fr
\and
 Changxiang He \\
 College of Science\\
   University of Shanghai\\
 for Science and Technology\\
 Shanghai 200093, China\\
 changxiang-he@163.com
 \and
 Yaping Mao \\
 Academy of Plateau Science and \\
 Sustainabilit and School of Mathematics\\
 and Statistics, Xining \\
 Qinghai 810008, China\\
 maoyaping@ymail.com}

\date{}

\runninghead{Ch. Yang et al.}{Perturbation Results for Distance-edge-monitoring Numbers}

\maketitle

\vspace*{-6mm}
\begin{abstract}
Foucaud {\em et al.}~recently introduced and initiated the study of a new graph-theoretic
concept in the area of network monitoring.
Given a graph $G=(V(G), E(G))$, a set $M \subseteq V(G)$ is a \emph{distance-edge-monitoring set} if for every edge $e \in E(G)$, there is a vertex $x \in M$ and a vertex $y \in V(G)$ such that the edge $e$ belongs to all shortest paths between $x$ and $y$.
The smallest size of such a set in $G$ is denoted by $\operatorname{dem}(G)$.
Denoted by $G-e$ (resp. $G \backslash u$) the subgraph of $G$ obtained by removing the edge $e$ from $G$ (resp. a vertex $u$ together with all its incident edges from $G$). In this paper, we first show that $\operatorname{dem}(G- e)-
\operatorname{dem}(G)\leq 2$ for any graph $G$
and edge $e \in E(G)$. Moreover, the bound is sharp.
Next, we construct two graphs $G$ and $H$ to show that $\operatorname{dem}(G)-\operatorname{dem}(G\setminus u)$ and
$\operatorname{dem}(H\setminus v)-\operatorname{dem}(H)$ can be arbitrarily large, where $u \in V(G)$ and $v \in V(H)$.
We also study the relation between
$\operatorname{dem}(H)$ and $\operatorname{dem}(G)$, where $H$ is a subgraph of $G$. In the end, we give an algorithm to judge whether the distance-edge-monitoring set still remain in the resulting graph when any edge of a graph $G$ is deleted.\\[2mm]
\smallskip
\end{abstract}

\section{Introduction}
In 2022, Foucaud {\it et~al.}~\cite{FKKMR21} introduced a new graph-theoretic concept called {\em distance-edge-monit\-oring set} (DEM for short), which means network monitoring using distance probes.
Networks are naturally modeled by finite undirected simple connected graphs, whose vertices represent
computers and whose edges represent connections between them. When a connection (an edge) fails in the network,
we can detect this failure, and thus achieve the purpose of monitoring the network.
Probes are made up of vertices we choose in the network. At any given moment, a probe of the network can measure its graph distance to
any other vertex of the network.
Whenever an edge of the network fails, one of the measured distances
changes, so the probes are able to detect the failure of any edge.
Probes that measure distances in graphs are present in real-life networks. They are useful in the fundamental
task of routing \cite{DABV06,GT00} and are also frequently used for problems concerning network verification \cite{BBDG15,BEEH06,BEMW10}.

In a network, we can put as few detectors as possible to monitor all the edges,
a natural question is whether the detectors placed in the original graph are still sufficient and need to be supplemented or reduced when some nodes or edges in the original graph are subjected to external interference and damage, we refer to \cite{Delen22,EROH15,Monson96,WEI22,Ye09}.
This kind of problem is usually called perturbation problem.

Graphs considered are finite, undirected and simple.
Let $G=(V(G),E(G))$ be a graph with vertex set $V(G)$ and edge set $E(G)$,
whose cardinality are
denoted by $|V(G)|$ and $e(G)$, respectively.
The \emph{neighborhood
set} of a vertex $v\in V(G)$ is $N_G(v)=\{u\in V(G)\,|\,uv\in
E(G)\}$.
Let $N_G[v]=N_G(v)\cup \{v\}$ be the  \emph{closed neighborhood
set  of a vertex $v$}.
The \emph{degree} of a vertex $v$ in $G$ is denoted
$d(v)=|N_{G}(v)|$. Let $\delta(G)$ and $\Delta(G)$ be the minimum and
maximum degree of a graph $G$, respectively.
For any subset $X$ of $V(G)$,
let $G[X]$
denote the subgraph of $G$ induced by $X$; similarly, for any subset
$F$ of $E(G)$, let $G[F]$ denote the subgraph induced by $F$.
We use
$G\setminus X$ to denote the subgraph of $G$ obtained by removing all the
vertices of $X$ together with the edges incident with them from $G$;
similarly, we use $G-F$ to denote the subgraph of $G$ obtained by removing all the edges of $F$ from $G$.
 If $X=\{v\}$ and
$F=\{e\}$, we simply write $G\setminus v$ and $G- e$ for $G-\{v\}$
and $G-\{e\}$, respectively.
For an edge $e$ of $G$, we denote by $G+e$ the graph obtained by adding an edge $e\in E(\overline{G})$ to $G$.
The {\it Cartesian product}\index{Cartesian product} $G\square H$ of
two graphs $G$ and $H$ is the graph whose vertex set is $V(G)\times
V(H)$ and whose edge set is the set of pairs $(u, v)(u',v')$ such
that either $uu'\in E(G)$ and $v=v'$, or $vv'\in E(H)$ and $u=u'$.
Let $G\vee H$ be a
\emph{join} graph of $G$ and $H$ with $V(G\vee H)=V(G)\cup V(H)$
and $E(G\vee H)=\{uv\,|\,u\in V(G),\,v\in V(H)\}\cup E(G)\cup E(H)$.
We denote by $d_G(x,y)$ the \emph{distance} between two vertices $x$ and $y$ in graph $G$.
For an edge $uv$ and a vertex $w\in V(G)$,
the distance between them is defined as
$d_G\left(uv, w\right)
=\min \{d_G\left(u, w\right),
d_G\left(v, w\right)\}$.
A $x$-$y$ path with length $d_G(x, y)$
in $G$ is a \emph{$x$-$y$ geodesic}.
Let $P_n$, $C_n$ and $K_n$ be the path, cycle and
complete graph of order $n$, respectively.

\subsection{DEM sets and numbers}

Foucaud et al.~\cite{FKKMR21} introduced a new graph-theoretic concept called DEM sets, which is relevant to network monitoring.
\begin{definition}\label{Defination:$P(M, e)$}
For a set $M$ of vertices and an edge $e$
of a graph $G$, let $P(M, e)$ be the set of pairs $(x, y)$ with a vertex $x$ of
$M$ and a vertex $y$ of $V(G)$ such that
 $d_G(x, y)\neq  d_{G- e}(x, y)$. In other words, $e$ belongs to all shortest paths between $x$ and $y$
in $G$.
\end{definition}

\begin{definition}
For a vertex $x$, let $EM(x)$ be the set of edges $e$ such that there exists a vertex $v$ in $G$ with $(x, v) \in  P(\{x\}, e)$, that is
$EM(x)=\{e\,|\,e \in E(G) \textrm{~and~ }
\exists v \in V(G)\textrm{~such that~}
d_G(x,v)\neq d_{G- e}(x,v)\}$
or $EM(x)=\{e\,|\,e \in E(G) \textrm{and }
P(\{x\}, e) \neq \emptyset \}$.
If $e \in EM(x)$,
we say that $e$ is monitored by $x$.
\end{definition}

Finding a particular vertex set $M$ and placing a detector on that set to monitor all edge sets
in $G$ have practical applications in sensor and network systems.
\begin{definition}
A vertex set $M$ of the graph $G$ is \emph{distance-edge-monitoring set} (DEM set for short) if every edge $e$ of $G$ is monitored by some vertex of $M$,
that is, the set $P(M, e)$ is nonempty. Equivalently, $\cup_{x\in M}EM(x)=E(G)$.
\end{definition}

\begin{theorem}{\upshape\cite{FKKMR21}}
\label{Th-Ncover}
Let $G $ be a connected graph with a vertex $x$ of $G$ and for any $y\in N(x)$, then,
we have $xy \in EM(x)$.
\end{theorem}

One may wonder to know the existence of such an edge detection set $M$. The
answer is affirmative. If
we take $M=V(G)$, then it follows
from Theorem \ref{Th-Ncover} that
$$
E(G) \subseteq  \cup_{x\in V(G)}
\cup_{y\in N(x)}\{ xy\}
 \subseteq \cup_{x\in V(G)}EM(x).
$$
Therefore, we consider the smallest cardinality of $M$ and give the following parameter.

\begin{definition}
The \emph{distance-edge-monitoring number} (DEM number for short) $\operatorname{dem}(G)$ of a graph $G$ is defined as the smallest size of a distance-edge-monitoring set of $G$, that is
$$
\operatorname{dem}(G)=\min\left\{|M|| \cup_{x\in M}EM(x)=E(G)\right\}.
$$
Furthermore, for any DEM set $M$ of $G$, $M$ is called a \emph{DEM
basis} if $|M|=\operatorname{dem}(G)$.
\end{definition}

The vertices
of $M$ represent distance probes in a network modeled by $G$. The DEM sets are very effective in network fault tolerance testing. For example, a DEM set can detect a failing edge, and
it can correctly locate the failing edge by distance from $x$ to $y$, because the
distance from $x$ to $y$ will increases when the edge $e$ fails.

Foucaud et al. \cite{FKKMR21} showed that
$1 \leq \operatorname{dem}(G) \leq n-1$ for any $G$ with order $n$, and graphs with $\operatorname{dem}(G)=1,n-1$ was characterized in \cite{FKKMR21}.

\begin{theorem}{\upshape\cite{FKKMR21}}
\label{th-dem-1}
Let $G$ be a connected graph with at least one edge. Then $\operatorname{dem}(G) = 1$ if and only if $G$ is a tree.
\end{theorem}

\begin{theorem}{\upshape\cite{FKKMR21}}
\label{th-dem-n}
$\operatorname{dem}(G) = n-1$ if and only if $G$ is the complete graph of order $n$.
\end{theorem}

\begin{theorem}{\upshape\cite{FKKMR21}}
\label{Th-forest}
For a vertex $x$ of a graph $G$, the set of edges $EM(x)$ induces a forest.
\end{theorem}
In a graph $G$,
the \emph{base graph $G_b$} of a graph $G$ is
the graph obtained from $G$ by iteratively
removing
vertices of degree $1$.
\begin{observation}{\upshape \cite{FKKMR21}}
\label{Obs:G_b}
Let $G$ be a graph  and $G_b$ be its base graph. Then we have $\operatorname{dem}(G) = \operatorname{dem}(G_b).$
\end{observation}

A  vertex set $M$  is called a \emph{vertex cover} of $G$ if
$M\cap \{u,v\}\neq \emptyset$ for  $uv\in E(G)$. The minimum cardinality of a vertex cover $M$ in $G$ is the \emph{vertex covering number} of $G$, denoted by $\beta(G)$.

\begin{theorem}{\upshape\cite{FKKMR21}}
\label{Theorem:Upperbond}
In any graph $G$ of order $n$, any vertex cover of $G$ is a DEM set of $G$, and thus $\operatorname{dem}(G) \leq \beta(G)$.
\end{theorem}

Ji et al.~\cite{JLKZ22} studied the Erd\H{o}s-Gallai-type problems for distance-edge-monitoring numbers. Yang et al.~\cite{Yang22}
obtained some upper and lower bounds of $P(M,e)$,
$EM(x)$, $\operatorname{dem}(G)$, respectively, and characterized the graphs with $\operatorname{dem}(G)=3$, and gave some properties of the graph $G$ with $\operatorname{dem}(G)=n-2$.
Yang et al.~\cite{YG24} determined the exact value of
distance-edge-monitoring numbers of grid-based pyramids, $M(t)$-graphs and
Sierpi\'{n}ski-type graphs.

\subsection{Progress and our results}

Perturbation problems in graph theory are as follows.

\begin{problem}\label{QP}
Let $G$ be a graph, and let $e\in E(G)$ and $v\in V(G)$. Let $f(G)$ be a graph parameter.

$(1)$ The relation between $f(G)$ and $f(G-e)$;

$(2)$ The relation between $f(G)$ and $f(G\setminus v)$.
\end{problem}

Chartrand et al.~\cite{Chart03} studied the perturbation problems on the metric dimension.
Monson et al.~\cite{Monson96} studied the effects of vertex deletion and edge deletion on the clique partition number in 1996. In 2015, Eroh et al.~\cite{EROH15}
considered the effect of vertex or edge deletion on the metric dimension of graphs.
Wei et al.~\cite{WEI22} gave some results on the edge metric dimension of graphs.
Delen et al.~\cite{Delen22} study the effect of vertex and edge deletion on the independence number of graphs.

A graph $H$ is a \emph{subgraph} of a graph $G$ if $V(H) \subseteq V(G)$ and $E(H) \subseteq E(G)$, in which case
 we write $H \sqsubseteq G$.
If $V(H)=V(G)$, then $H$ is a \emph{spanning subgraph} of $G$.
If $H$ is a subgraph of a graph $G$, where $H \neq G$, then $H$ is a \emph{proper subgraph} of $G$. Therefore, if $H$ is a proper subgraph of $G$, then either $V(H)\subset V(G)$ or $E(H)\subset E(G)$.

\medskip
We first consider the existence of graphs with given values of DEM numbers.
\begin{problem}\label{Qst}
Let $r,s,n$ be three integers with $1 \leq r,s
\leq n-1$.

$(1)$ Is there a connected graph $G$ of order $n$ such that $\operatorname{dem}(G)=r$?

$(2)$ Let $G$ be a connected graph of order $n$.
Is there a connected subgraph $H$ in $G$
such that $\operatorname{dem}(H)=s$ and $\operatorname{dem}(G)=r$?
\end{problem}

In Section $2$, we give the answers to Problem \ref{Qst}.
\begin{proposition}\label{Obs:EST}
For any two integers $r, n$ with $1 \leq r \leq n-1$, there exists  a connected graph $G$ of order $n$ such that $\operatorname{dem}(G)=r$.
\end{proposition}

\begin{corollary}\label{cor:ESTC}
Given three integers $s, t, n$ with $1 \leq s \leq t \leq n-1$,
there exists a connected graph $H\sqsubseteq G$ such that $\operatorname{dem}(H)=s$ and $\operatorname{dem}(G)=t$.
\end{corollary}

In Section $3$, we focus on Problem \ref{QP} $(1)$ and study the difference between $\operatorname{dem}(G-e)$ and $\operatorname{dem}(G)$.
\begin{theorem}\label{th-Difference}
Let $G$ be a graph. For any edge $e \in E(G)$, we have
$$
\operatorname{dem}(G-e)-\operatorname{dem}(G) \leq 2.
$$
Moreover, this bound is sharp.
\end{theorem}

Let $G$ be a graph and $E\subseteq E(\overline{G})$. Denote by $G+E$ the graph with $V(G+E)=V(G)$ and $E(G+E)=E(G)\cup E$.
We construct graphs with the following properties in Section $3$.
\begin{theorem}\label{th-Ei}
For any positive integer $k\geq 2$, there exists
a graph sequence  $\{G^i\,|\,0\leq i\leq k \}$,
with $e(G^i)-e(G^0)=i$ and $V(G^i)=V(G^j)$ for $0\leq i,j \leq k$,
such that
$\operatorname{dem}(G^{i+1})
-\operatorname{dem}(G^0)=i$, where $1\leq i\leq k-1$.
Furthermore,  we have $\operatorname{dem}(G^0)=1$,
$\operatorname{dem}(G^1)=2$ and $\operatorname{dem}(G^i)=i$, where $2\leq i\leq k$.
\end{theorem}

A \emph{feedback edge set} of a graph $G$ is a
set of edges such that removing them from
$G$ leaves a forest. The smallest size
of a feedback edge set of $G$ is
denoted by $\operatorname{fes}(G)$
(it is sometimes called the
cyclomatic number of $G$).
\begin{theorem}{\upshape\cite{FKKMR21}}
\label{Th-fes}
If $\operatorname{fes}(G) \leq 2$,
then $\operatorname{dem}(G)
\leq \operatorname{fes}(G)+1$.
Moreover, if $\operatorname{fes}(G) \leq 1$,
then equality holds.
\end{theorem}

Theorem \ref{Th-fes}
implies the following corollary, and its proof will be given in Section $3$.
\begin{corollary}\label{cor-e}
Let $T_n$ be a tree of order $n$, where $n\geq 6$. For  edges
$e_1,e_2\in E(\overline{T_n})$, we have

$(1)$ $\operatorname{dem}(T_n+e_1)=\operatorname{dem}(T_n)+1$.

$(2)$ $\operatorname{dem}(T_n+\{e_1,e_2\})=2$ or $3$.
\end{corollary}

The following result shows that there exists a graph $G$ and an induced subgraph $H$ such that the difference
$\operatorname{dem}(G)-\operatorname{dem}(H)$ can be arbitrarily large; see Section 4 for proof details. In addition, we also give an answer to the Problem \ref{QP} $(2)$.
\begin{theorem}\label{Obs:dv1}
For any positive integer $k$, there exist two graphs $G_1,G_2$ and their non-spanning subgraphs $H_1,H_2$
such that
$$
\operatorname{dem}(G_1)-\operatorname{dem}(H_1)=k \ and \ \operatorname{dem}(H_2)-\operatorname{dem}(G_2)=k.
$$
\end{theorem}
Furthermore, $\operatorname{dem}(G)-\operatorname{dem}(H)$ can be arbitrarily large, even for $H=G\setminus v$.
\begin{theorem}\label{TH:deEV}
For any positive integer $k$, there exist two graphs $G,H$ and two vertices $u\in V(G)$,
$v\in V(H)$ such that

$(1)$ $\operatorname{dem}(G)
-\operatorname{dem}(G\setminus u)\geq k$;

$(2)$ $\operatorname{dem}(H\setminus v)
-\operatorname{dem}(H)\geq k$.
\end{theorem}

For a connected graph $G$ of order $n$, where $n$ is fixed, the difference between $\operatorname{dem}(G)$ and $\operatorname{dem}(G\setminus v)$ can be bounded.

\begin{proposition}\label{pro-upper}
For a connected graph $G$ with order $n \ (n\!\geq\! 2)$ and $v\!\in\! V(G)$,
if $G\setminus v$ contains at least one edge, then
$\operatorname{dem}(G)-\operatorname{dem}(G\setminus v)\! \leq n-2$.
Moreover, the equality holds if and only if $G$ is$\;K_3$.
\end{proposition}

\begin{theorem}\label{th-dem-2}
Let $G$ be a connected graph with
order $n\geq 4$ and
$\operatorname{dem}(G) = 2$.
Let $E\subseteq E(G)$.
If $\operatorname{dem}(G)=
\operatorname{dem}(G-E)$,
then $|E| \leq 2n-6$.
Furthermore, the bound is sharp.
\end{theorem}

For $H\sqsubseteq G$, the \emph{DEM set of $H$ in $G$}
is a set $M\subseteq V(H)$ such that $E(H) \subseteq \bigcup\limits_{x\in M}EM(x)$.
\begin{definition}
For $H\sqsubseteq G$,
the \emph{restrict-DEM number} $\operatorname{dem}(G|_H)$ of a graph $G$ is defined as the smallest size of a DEM set of $H$ in $G$, that is,
$$
\operatorname{dem}(G|_H)=\min\left\{|M|\Big| E(H) \subseteq \cup_{x\in M} EM(x),
M\subseteq V(H)\right\}.
$$
\end{definition}

\begin{figure}[!h]
\vspace*{-5mm}
\centering
\includegraphics[width=7cm]{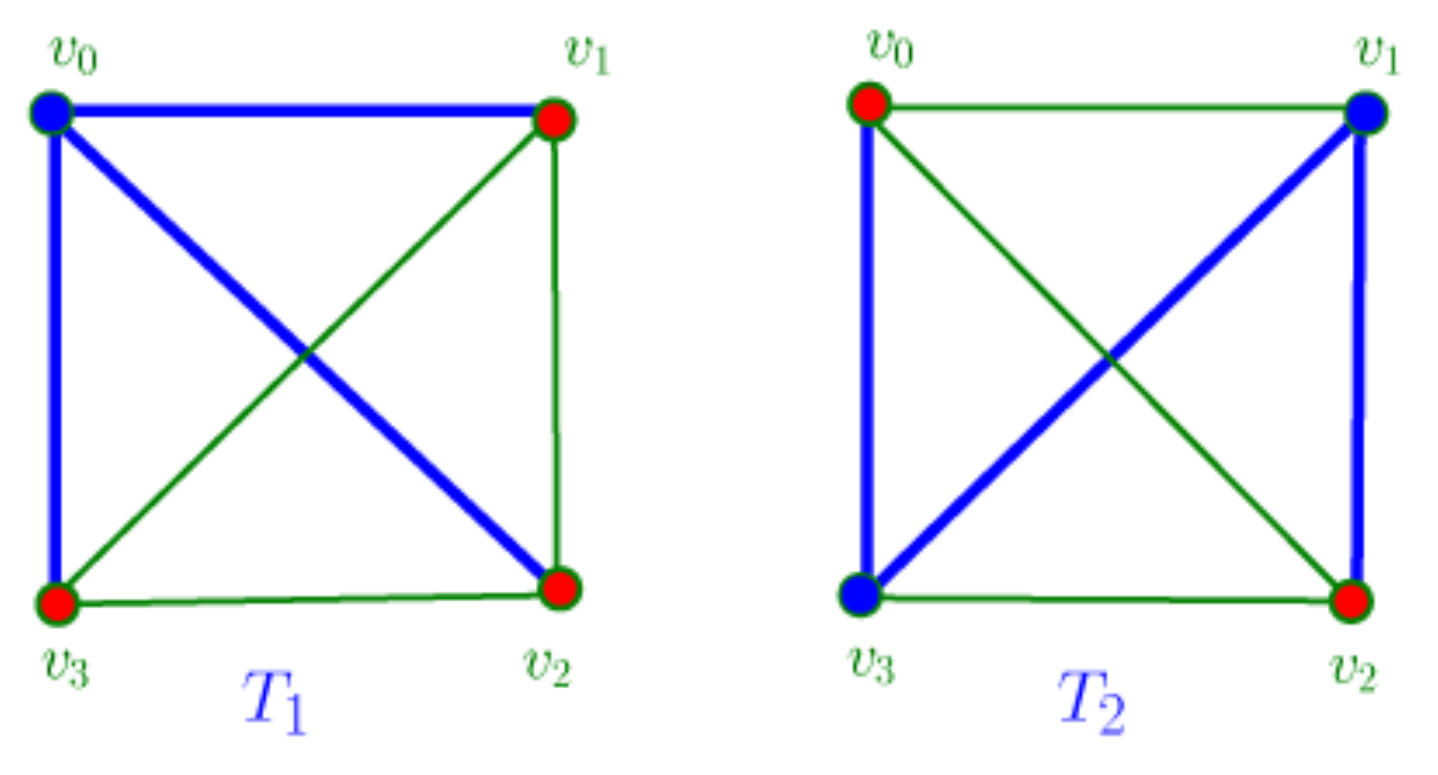}\vspace*{-5mm}
\caption{The blue edges are those of trees $T_1$ and $T_2$ in $K_4$.}
\label{Fig:Tree}\vspace*{-3mm}
\end{figure}

\begin{example}
Let $G=K_4$ with $V(G)=\{v_0, v_1,v_2,v_3\}$ and $E(G)=\{v_iv_j\,|\,0\leq i<j\leq 3\}$.
Let $T_1$ and $T_2$ be the subgraphs of $G$ with $E(T_1)=\{v_0v_1,v_0v_2,v_0v_3\}$
and $E(T_2)=\{v_0v_3,v_3v_1,v_1v_2\}$.
Then, $\operatorname{dem}(K_4|_{T_1})=1$ and $\operatorname{dem}(K_4|_{T_2})=2$.
The DEM set of subgraph $T_i$ ($i=1,2$) in $K_4$ is shown in
Figure~\ref{Fig:Tree},
where the blue vertices form the set $M$.
The reason as follows.\\
Let $M_1=\{v_0\}$.
Since $v_0v_1,v_0v_2,v_0v_3 \in EM(v_0)$,
it follows that
$\operatorname{dem}(K_4|_{T_1})\leq 1$.
Obviously, $\operatorname{dem}(K_4|_{T_1})\geq 1$, and hence $\operatorname{dem}(K_4|_{T_1})=1$.
Then, we prove that $\operatorname{dem}(K_4|_{T_2})=2$.
Since $d_G(v_0,v_1)=d_{G-v_1v_2}(v_0,v_1)=1$ and
$d_G(v_0,v_2)=d_{G-v_1v_2}(v_0,v_2)=1$,
it follows that $v_1v_2\notin EM(v_0)$.
Similarly, $v_1v_3\notin  EM(v_0)$.
Therefore, $v_1v_2,v_1v_3 \notin  EM(v_0)$.
By a similar argument, we have
$v_0v_3\notin  EM(v_1)$,
$v_1v_3,v_0v_3\notin  EM(v_2)$ and
$v_1v_2\notin  EM(v_3)$, and hence
$\operatorname{dem}(K_4|_{T_2})\geq 2$.
Let $M=\{v_1,v_3\}$. Then,
$v_1v_2,v_1v_3\in  EM(v_1)$,
$v_1v_3,v_0v_3\in  EM(v_3)$, and hence $\operatorname{dem}(K_4|_{T_1})\leq 2$.
Therefore, we have $\operatorname{dem}(K_4|_{T_2})=2$, and so $\operatorname{dem}(K_4|_{T_i})=i$ ($i=1,2$).
\end{example}

\begin{theorem}\label{The:sTN}
Let $T$  be a spanning
tree of $K_n$. Then
$1 \leq \operatorname{dem}(K_n|_T) \leq
\lfloor n/2\rfloor.$
Furthermore, the bound is sharp.
\end{theorem}

In Section $5$, we focus on the following problem and give an algorithm
to judge whether the DEM set is still valid in the resulting graph when any edge (or vertex) of a graph $G$ is deleted.
\begin{problem}\label{Q4}
For any graph $G$, if some edges or vertices in $G$ is deleted, we want to
know whether the original DEM set can monitor all edges.
\end{problem}

\section{Results for Problem 2}

A \emph{kite} $K(r, n)$ is a graph obtained from the
complete graph $K_{r+1}$ and a path $P_{n-r}$ by attaching a vertex of $K_{r+1}$ and one end-vertex of $P_{n-r}$;
see an example of $K{(7, 12)}$ in Figure \ref{K59}.

\begin{figure}[!htbp]
 \centering
  \includegraphics[width=8.42cm]{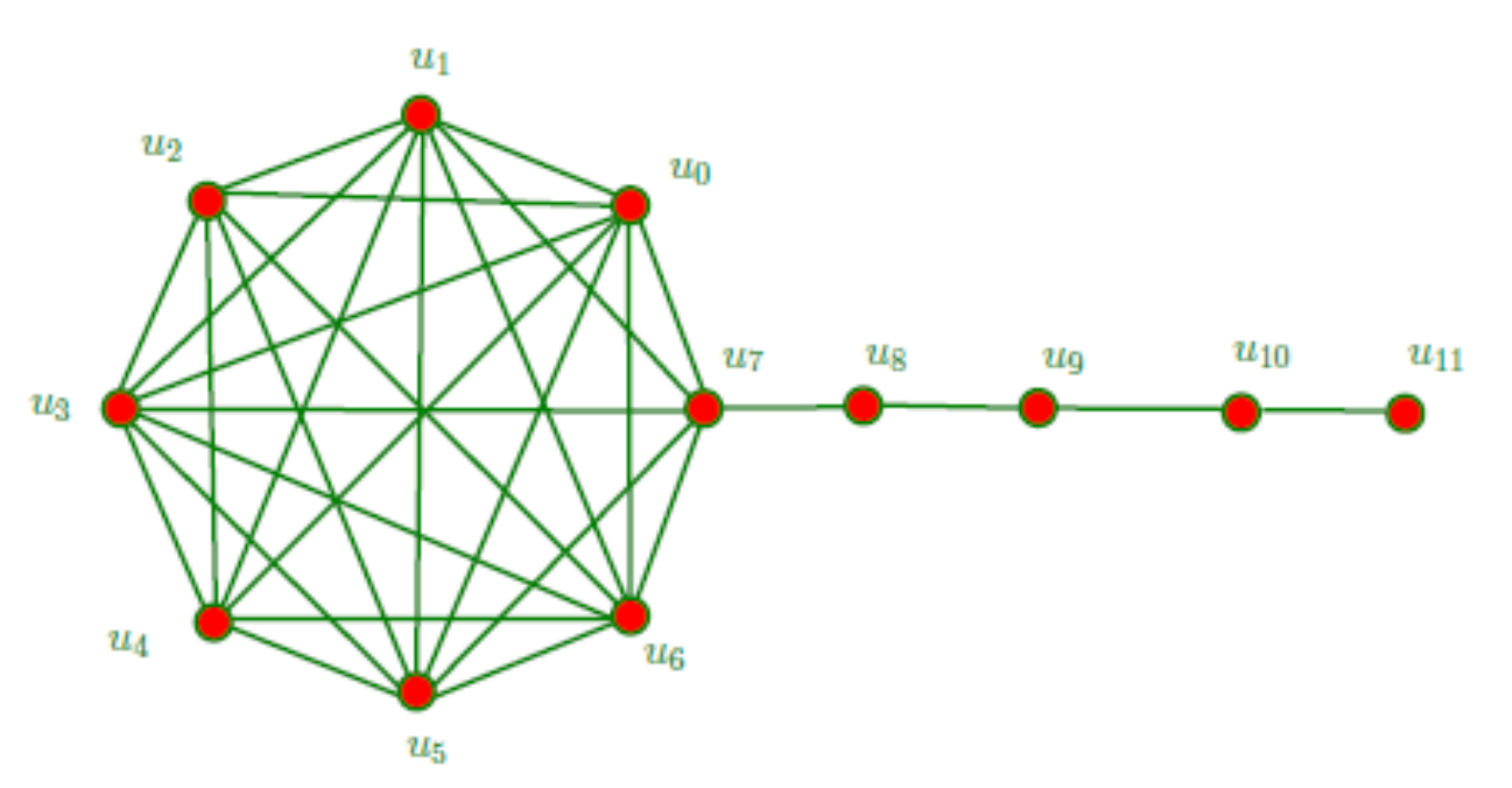}\vspace*{-4mm}
  \caption{The graph $K{(7, 12)}$}
  \label{K59}
\end{figure}

We first give the proof of Proposition \ref{Obs:EST}. \\[0.1cm]

\noindent {\bf Proof of Proposition \ref{Obs:EST}:}
Let $G=K{(r, n)}$ with $V(G)=\{u_i\,|\,0\leq i\leq n-1\}$ and
$E(G)=\{u_iu_{j}\,|\,0\leq i<j\leq r\}$
$\cup \{u_{r+s}u_{r+s+1}\,|\,0 \leq  s\leq n-r-2\}$. From Observation  \ref{Obs:G_b} and
Theorem \ref{th-dem-n},
we have $\operatorname{dem}(G)=
\operatorname{dem}(G_b)=
\operatorname{dem}(K_{r+1})=r.$
In fact, for the above $G$, the path $P_{n-r-1}$
can be replaced by $T_{n-r-1}$, where $T_{n-r-1}$ is any
tree of order $n-r-1$.\QED \medskip

Proposition \ref{Obs:EST} shows that Corollary \ref{cor:ESTC}
is true.
For three integers $s, t, n$ with
$1 \leq s \leq t \leq n-1$,
let $G=K{(t, n)}$ and $H=K{(s, n)}\sqsubseteq G$.
From Proposition \ref{Obs:EST},
$\operatorname{dem}(G)=t$ and $\operatorname{dem}(H)=s$.
Therefore, there exists a connected graph
$H\sqsubseteq G$ such that $\operatorname{dem}(H)=s$ and $\operatorname{dem}(G)=t$.

This gives an answer about Problem
\ref{Qst}, see  Corollary \ref{cor:ESTC}.
One might guess that if $H$ is a subgraph of $G$, then $\operatorname{dem}(H)\leq \operatorname{dem}(G)$, however we will show in the next section that there is no monotonicity for
the DEM number.

\section{The effect of deleted edge}

The following observation is immediate.
\begin{observation}\label{Obs:disjoint}
Let $G_1,G_2,...,G_m$ be the connected components of $G$. Then
$$\operatorname{dem}(G)=
\operatorname{dem}\left(G_1\right)+\cdots
+\operatorname{dem}\left(G_m\right).$$
\end{observation}

Furthermore, we suppose that the DEM number of $K_1$ is $0$.

\begin{proposition}\label{Obs:CUV}
For any $uv\in E(G)$,
$uv \notin EM(w)$ for $w \in
\left(N_G(u)\cup N_G(v)\right)\setminus \{u,v\}$
if and only if
$uv$ is only monitored by $u$ and $v$.
\end{proposition}

\begin{proof}
Since $w \in \left(N_G(u)\cup N_G(v)\right)\setminus \{u,v\}$ and $uv \notin EM(w)$, it follows that
$d_G(w,u)=d_{G - uv}(w,u)$ and $d_G(w,v)=d_{G- uv}(w,v)$.
For any $x \in V(G)- N_G[u]\cup N_G[v]$,
the path from $x$ to $u$ must through
$w_1$, where $w_1 \in \left(N_G(u)\cup N_G(v)\right)\setminus \{u,v\}$.
Then
$d_G(x,u)=d_G(x,w_1)+d_G(w_1,u)=
d_G(x,w_1)+d_{G-uv}(w_1,u)=
d_{G-uv}(x,w_1)+d_{G-uv}(w_1,u)=
d_{G-uv}(x,u)$.
Similarly, $d_G(x,v)=d_{G-uv}(x,v)$.
For any $x\in V(G)-\{u,v\}$,
we have
$uv\notin EM(x)$.
From Theorem
\ref{Th-Ncover},
$uv\in EM(u)$ and $uv\in EM(v)$,
and hence $uv$ is only monitored by the vertex in
$\{u,v\}$.

\smallskip
Conversely, if
$uv$ is only monitored by $u$ and $v$,
then $uv\notin EM(w)$ for any $w \in V(G)\setminus\{u,v\}$,
Especially, since $\left(N_G(u)\cup N_G(v)\right)\setminus \{u,v\} \subseteq V(G)\setminus\{u,v\}$,
it follows that $uv \notin EM(w)$ for $w \in
\left(N_G(u)\cup N_G(v)\right)\setminus \{u,v\}$, as desired.
\end{proof}

Then, we give the proof of Theorem \ref{th-Difference}. \\

\noindent {\bf Proof of Theorem \ref{th-Difference}:}
If $G$ is a disconnected graph, then the edge $e$ must be in some connected component $G_1$ of $G$ for any $e\in E(G)$, and hence
$e$ can only be monitored by the vertex in $V(G_1)$.
Therefore, we just need consider the graph $G$ which is connected.
Let $M$ be a DEM set of $G$ with
$|M|=\operatorname{dem}(G)$
and $e=uv\in E(G)$.
If $M$ is also a
DEM set of $G-e$,
then $\operatorname{dem}(G-e)
\leq \operatorname{dem}(G)$.
Otherwise, let $M^{\prime}=M \cup\{u, v\}$.
It suffices to show that $M'$
is a DEM set of $G-e$.

\medskip
If $G-e$ has two components, say $G_1$ and $G_2$, then $e$ is a cut edge of $G$ and from Observation \ref{Obs:disjoint}, we have $\operatorname{dem}(G-e)=\operatorname{dem}\left(G_1\right)+\operatorname{dem}\left(G_2\right)$. Without loss of generality, assume that
$u\in V\left(G_1\right)$ and $v\in V\left(G_2\right)$.

\begin{fact}\label{fact1}
$\operatorname{dem}\left(G_1\right)
\leq\left|\left(M \cap V\left(G_1\right)\right)\cup\{u\}\right|$ and $\operatorname{dem}\left(G_2\right)
\leq\left|\left(M \cap V\left(G_2\right)\right)
\cup\{v\}\right|$.
\end{fact}
\begin{proof}
For any edge $e_1=x_1y_1 \in E\left(G_1\right)$,
if there exists a vertex $w \in V\left(G_1\right)\cap M$ such that
$e_1\in EM(w)$, then we are done.
Otherwise,
there exists a vertex
$w \in V\left(G_2\right)\cap M$
such that $d_{G-e_1}\left(x_1, w\right)
\neq d_G\left(x_1, w\right)$ or
$d_{G-e_1}\left(y_1, w\right)
\neq d_G\left(y_1, w\right)$.
Without loss  of generality,
we suppose that $d_{G-e_1}\left(y_1, w\right)
\neq d_G\left(y_1, w\right)$ and
$d_G\left(w, e_1\right)=d_G\left(w,x_1\right)$.
Since
$d_G\left(y_1, w\right)
=d_G\left(y_1, x_1\right)+
d_G\left(x_1, u\right)
+d_G(u, w)$,
$d_{G- \{ e, e_1\}}\left(x_1, u\right)= d_{G-e_1}\left(x_1, u\right)$
and $d_{G- \{ e, e_1\}}\left(y_1, x_1\right)>
d_{G-e}\left(y_1, x_1\right)$,
it follows that
$$
\begin{aligned}
d_{G- \{ e, e_1\}}\left(u, y_1\right)
=&d_{G- \{ e, e_1\}}\left(u, x_1\right)+d_{G- \{ e, e_1\}}\left(x_1, y_1\right)\\
=&d_{G- \{ e, e_1\}}\left(u, x_1\right)+d_{G-e}\left(x_1, y_1\right)\\
>&d_{G- e}\left(u, x_1\right)+d_{G- e}\left(x_1, y_1\right)\\
=&d_{G- e}\left(u, y_1\right)
\end{aligned}
$$
and hence $d_{G- \{ e, e_1\}}\left(y_1, u\right)
\neq d_{G- e_1}\left(y_1, u\right)$.
Therefore, $e_1$ is monitored by $\left(M \cap V\left(G_1\right)\right)\cup\{u\}$
in graph $G-e$.
This implies that
$\operatorname{dem}\left(G_1\right)
\leq\left|\left(M \cap V\left(G_1\right)\right)\cup\{u\}\right|$.
Similarly, we can obtain that
$\operatorname{dem}\left(G_2\right)
\leq\left|\left(M \cap V\left(G_2\right)\right)
\cup\{v\}\right|$.
\end{proof}
From Fact \ref{fact1}, we have $\operatorname{dem}(G- e)\leq\left|M^{\prime}\right|=
\left|M \cup\{u, v\}\right|
\leq\left|M\right|+2=
\operatorname{dem}(G)+2$.
\eject

Suppose that $G-e$ is connected.
If $M$ is also a
DEM set of $G- e$,
then $\operatorname{dem}(G-e)
\leq |M|=\operatorname{dem}(G)$ and
we are done.
Otherwise,
there exists $e_1=x y \in E(G- e)$ such
that the edge $e_1$ is not monitored by $M$ in $G- e$.
Since $M$ is a distance- edge-monitoring
set of $G$, it follows that
there exists a vertex $z \in M$ such that
$d_{G- e_1}(x, z) \neq d_G(x, z )$
or $d_{G- e_1}(y, z) \neq d_G(y, z)$.
In addition, since $e_1$ is not monitored
by $M$ in $G- e$, it follows that
the distance from $z$ to $x$ or
$y$ is not changed
after removing the edge $e_1$ in $G- e$, which means that
$d_{G- \{ e, e_1\}}\left(y, z\right)
=d_{G- e}\left(y, z\right)$ and
$d_{G- \{ e, e_1\}}\left(x, z\right)
=d_{G- e}\left(x, z\right)$.
If
$d_G\left(e_1, z\right)
=d_G(x, z)$, then the edge $e$ lies on
every $z-y$ geodesic in $G$ for
$z\in M$ and $xy\in EM(z)$ in $G$,
otherwise there exists $z^*\in M$
and  $xy\in EM(z^*)$
such that $e$ does not appear in $z^*-y$ geodesic
in $G$,
that is $d_{G- e}\left(x, z^*\right)
=d_G\left(x, z^*\right)$
and
$d_{G- \{ e, e_1\}}\left(x, z^*\right)
\neq d_G\left(x, z^*\right)$,
which contradicts to
the fact that $M$ is not the
DEM set of
graph $G-e$.
\begin{claim}\label{claim2}
If a geodesic in $G$ from $z$ to $y$ traverses
the edge $e$ in the order $u, v$,
then each geodesic in $G$
from $z$ to $y$ traverses $e$ in the order $u,v$.
\end{claim}

\begin{proof}
Assume, to the contrary, that there exists
two $z-y$ geodesics $P^g_1$ and $P^g_2$,
where $P^g_1=z \ldots  u  v \ldots y$ and
$P^g_2=z \ldots vu\ldots y$.
The $z-y$ geodesic $P^g_1$ implies that
$d(u, v)+d(v, y)=d(u, y)$,
 and the $z-y$ geodesic $P^g_2$ implies that
$d(v, u)+d(u, y)=d(v, y)$,
and hence $d(u, v)=0$, a contradiction.
\end{proof}

From Claim \ref{claim2}, without loss of generality, we may assume that
every geodesic in $G$ from $z$ to $y$ traverses
the edge $e$ in the order $u, v$.
Thus, we have $d_G(z, y)=d_G(z, v)+d_G(v, y)$.
We now show that $xy$ can be
monitored by $v$ in $G- e$.
Note that
$d_{G- e_1}(z, y) \neq d_{G}(z, y)$,
$d_{G- e}(v, y)=d_{G}(v, y)$
and $d_{G- e}(x, y)=d_{G}(x, y)$.
Then
$d_{G- \{ e, e_1\}}\left(v, y\right)=$
$d_{G- \{ e, e_1\}}\left(v, x\right)+$
$d_{G- \{ e, e_1\}}\left(x, y\right)$
$=d_{G- e_1}\left(v, x\right)+$
$d_{G- e_1}\left(x, y\right)$
$>d_{G}\left(v, x\right)+ $
$d_{G}\left(x, y\right)$
$=d_{G- e}\left(v, x\right)+$
$d_{G- e}\left(x, y\right)
\geq d_{G- e}(v, y)$.
Since $d_{G- e}(v, y)
> d_{G- \{ e, e_1\}}(v, y)$,
it follows that $e_1$ can be monitored by $v$.
Since $e_1\in EM(u)$ or $e_1\in EM(v)$, it follows that
$M^{\prime}=M \cup\{u, v\}$ is
a distance edge-monitoring-set of $G- e$,
and thus $\operatorname{dem}(G- e)
\leq \operatorname{dem}(G)+2$,
as desired. \QED \smallskip

Li et al.~\cite{weli22} got the following result about
DEM numbers of
$ C_k\square P_{\ell}$.
\begin{theorem}{\upshape \cite{weli22}}
\label{ThmCnPn}
Let $\ell$ and $k$ be two integers with $\ell \geq 3$ and $k \geq 2$. Then
$$
\operatorname{dem}\left(C_k \square P_{\ell}\right)=
\begin{cases}k & \text { if } k \geq 2 \ell+1, \\ 2\ell & \text { if } k<2 \ell+1.\end{cases}
$$
\end{theorem}

To show the sharpness of Theorem \ref{th-Difference}, we consider the following proposition.
\begin{proposition}\label{Lem:eq2}
There exist two connected graphs $G_1,G_2$ of order $n$ such that $\operatorname{dem}(G_1- e)-\operatorname{dem}(G_1)
=2$ and
$\operatorname{dem}(G_2)-\operatorname{dem}(G_2- e)=2$.
\end{proposition}

\begin{proof}
Firstly, we consider the graph $G_{1} \ (|V(G_1)|=n\geq 8)$ with vertex set $V(G_{1})=\{v_i|1\leq i\leq n-8\}
\cup \{u_i|1\leq i\leq 8\}$
and edge set $E(G_{1})=\{u_iv_i\,|\,1\leq i\leq 8\}
\cup \{u_iu_{i+1}\,|\,1\leq i\leq 7\}
\cup \{v_iv_{i+1}\,|\,1\leq i\leq 7\}
\cup \{u_1u_{8}\} \cup \{u_1u_{5}\}
\cup \{v_1v_{8}\}
\cup \{v_1v_{9}\}
\cup \{v_iv_{i+1}\,|\,9\leq i\leq n-9\}$.
Let $G^*_{8}=G_b(G_1)$.
 Obviously, $G^*_{8}$ is the base graph of
$G_1$, which is obtained  by removing
 the all edge in the edge set $\{v_1v_{9}\}
\cup \{v_iv_{i+1}\,|\,9\leq i\leq n-9\}$.
The graphs $G^*_{8}$ and $G^*_{8}-u_1u_5$ are shown in Figures \ref{Fig:G_8} and \ref{Fig:G_81},
respectively.

\begin{figure}[!htbp]
\vspace*{-4mm}
\centering
\begin{minipage}{0.45\linewidth}
\vspace{3pt}
\centerline{\includegraphics[width=5.5cm]{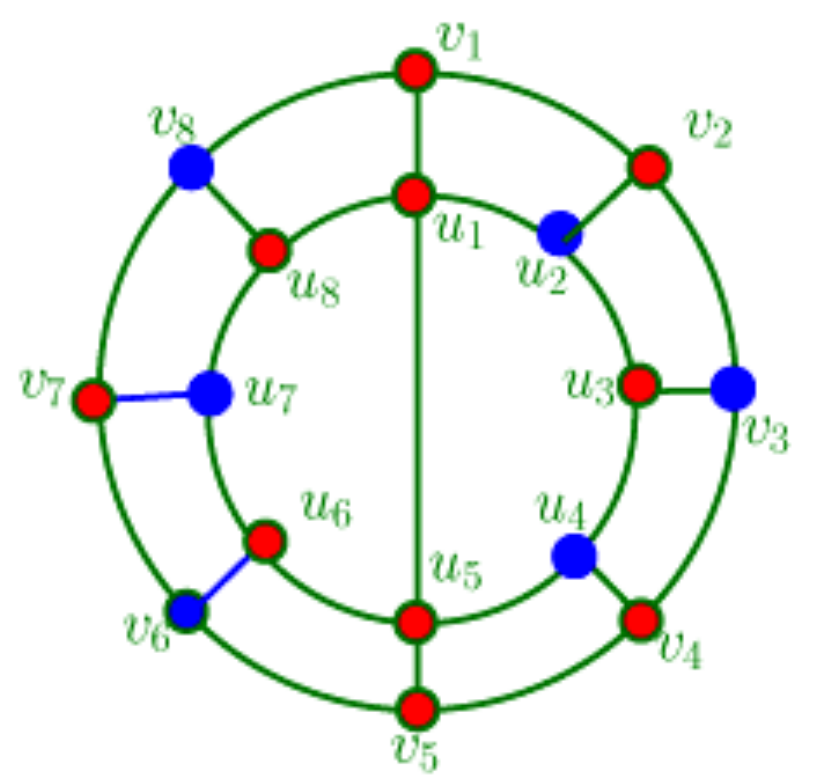}}\vspace*{-2mm}
\caption{$\operatorname{dem}(G_8^*)=6$}
\label{Fig:G_8}
\end{minipage}
\begin{minipage}{0.45\linewidth}
\vspace{3pt}
\centerline{\includegraphics[width=5.9cm]{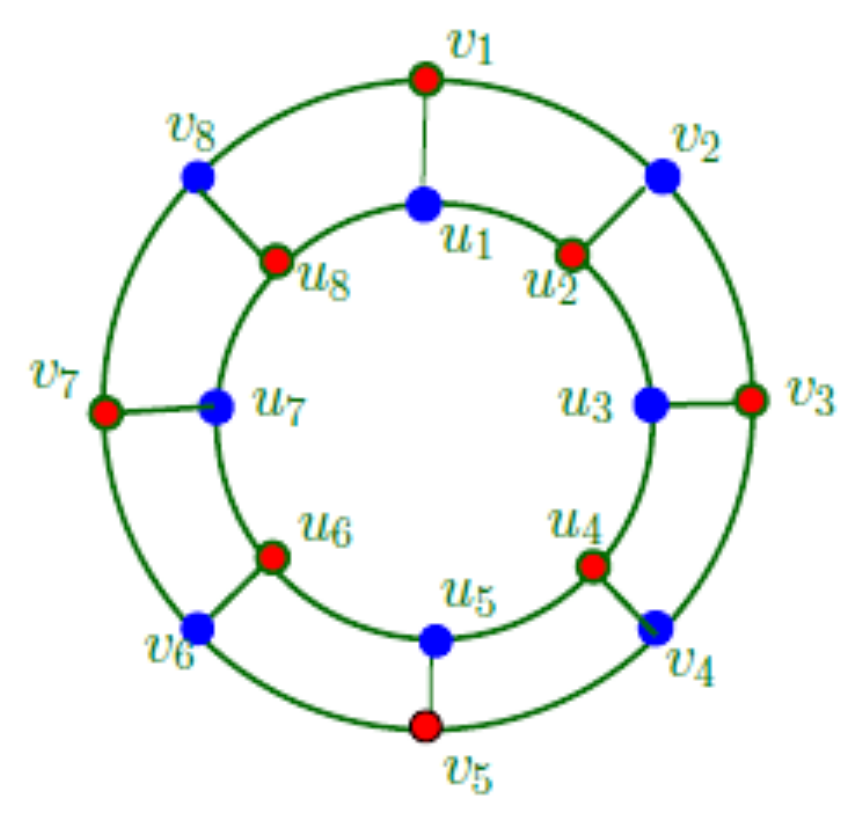}}\vspace*{-5mm}
\caption{$\operatorname{dem}(G_8^*-u_1u_5)=8$}
\label{Fig:G_81}
\end{minipage}
\end{figure}

\medskip
Let $M_1=\{u_2,u_4, v_3,v_6,u_7,v_8\}$.
Note that $\{u_1u_5,u_5v_5,u_2v_2$,
$u_2u_1,u_2u_3 \}\subseteq EM(u_2)$,
$\{v_1u_1,$ $u_4u_3,u_4u_5,u_4v_4\}\subseteq EM(u_4)$,
$\{v_3u_3,v_2v_3,v_4v_3,v_5v_4,v_2v_1\}$
$\subseteq EM(v_3)$,
$\{v_8v_1,u_8v_8,$
$v_8v_7\}$
$\subseteq  EM(v_8)$,
$\{u_7u_8,u_8u_1,u_{6}u_7,u_{6}u_5,u_7v_7\}$
$\subseteq  EM(u_7)$ and
$\{v_5v_6,v_6v_7,u_{6}v_6\}\in EM(v_6)$.
Therefore, $E(G_8^*)=\cup_{x\in M_1}EM(x)$,
and hence
$\operatorname{dem}(G_8^*)\leq |M_1|=6$.

\medskip
Let $M$ be a DEM set of $G^*_8$ with  the  minimum  cardinality.
For the edge $u_iv_i$, where $2\leq i\leq 8$ and $i\neq 5$,
and any $w\in (N(u_i)\cup N(v_i))\setminus\{u_i,v_i\}$,
we have
$d_{G- u_iv_i}(w,u_i)=d_{G}(w,u_i)$
and $d_{G- u_iv_i}(w,v_i)=d_{G}(w,v_i)$,
and hence $u_iv_i \notin EM(w)$.
From Proposition \ref{Obs:CUV}, the edge
$u_iv_i$ ($2\leq i\leq 8$ and $i\neq 5$)
is only monitored by $\{u_i, v_i\}$,
and hence $M\cap\{u_i, v_i\}\neq \emptyset$
for $2\leq i\leq 8$ and $i\neq 5$, and so $\operatorname{dem}(G^*_8)\geq 6$.
Therefore, $\operatorname{dem}(G^*_8)=6$.

\medskip
Since $G^*_8- u_1u_5\cong C_8\square P_2$,
it follows from Theorem \ref{ThmCnPn} that $\operatorname{dem}(G^*_8 - u_1u_5)=
\operatorname{dem}(C_8\square P_2)=8$.
From Observation \ref{Obs:G_b},
$\operatorname{dem}(G_1 - u_1u_5)
-\operatorname{dem}(G_1)$
$=\operatorname{dem}(G^*_8- u_1u_5)
-\operatorname{dem}(G^*_8)=$
$8-6=2$, as desired.

\begin{figure}[!b]
\vspace*{-4mm}
\centering
\begin{minipage}{0.45\linewidth}
\vspace{3pt}
\centerline{\includegraphics[width=4.9cm]{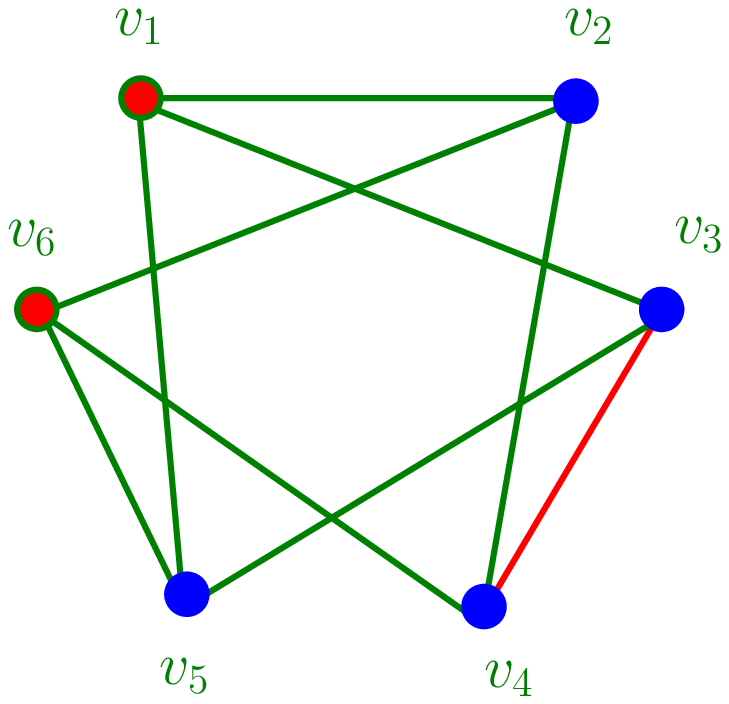}}\vspace*{-3mm}
\caption{$\operatorname{dem}(G^{\prime}_6)=4$}
\label{Fig:DemG_4}
\end{minipage}
\begin{minipage}{0.45\linewidth}
\vspace{2pt}
\centerline{\includegraphics[width=5cm]{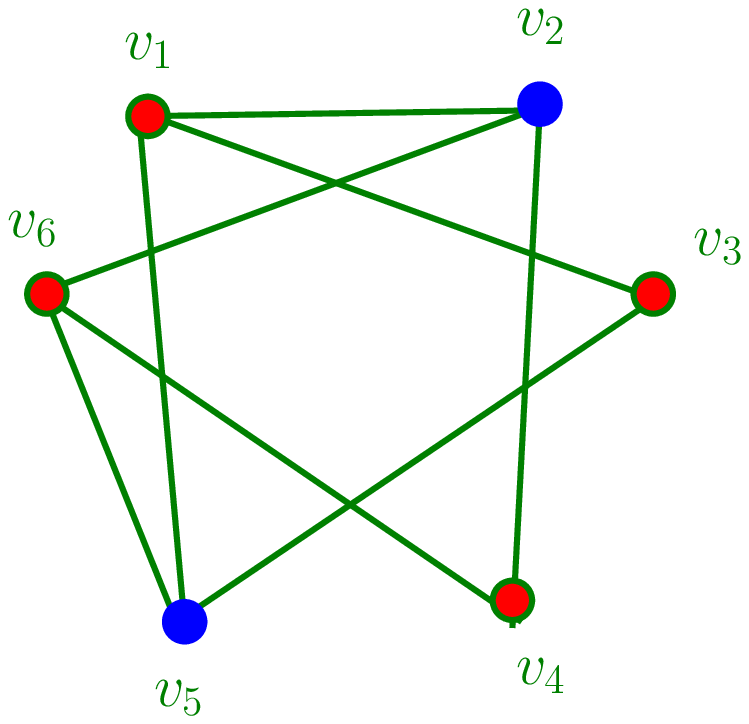}}\vspace*{-3mm}
\caption{$\operatorname{dem}(G^{\prime}_6-v_3v_4)=2$}
\label{Fig:DemG_2}
\end{minipage}
\end{figure}

Next, we consider the graph $G_2 \ (|V(G_2)|=n\geq 6)$ with vertex set $V(G_{2})=\{v_i\ |\ 1\leq i \leq n\}$
and edge set $E(G_{2})=\{v_{1}v_{2}, v_{3}v_{4}$,
$v_{5}v_{6},v_{1}v_{3},v_{1}v_{5}$,
$v_{2}v_{4},v_{2}v_{6}$,
$v_{3}v_{5},v_{4}v_{6}\}
\cup\{v_iv_{i+1}\ |\ 6\leq i\leq n-1\}$.
Let $G^{\prime}_{6}$ be the base graph
of $G_2$, that is, $G_b(G_2)=G^{\prime}_{6}$.
The graphs $G^{\prime}_{6}$ and $G^{\prime}_{6}-v_1v_3$,
are shown in Figure \ref{Fig:DemG_4} and
Figure \ref{Fig:DemG_2}, respectively.
From Observation \ref{Obs:G_b}, $\operatorname{dem}(G_2)
=\operatorname{dem}(G^{\prime}_{6})$.

\medskip
Take $M'_1=\{v_2,v_3, v_4,v_5\}$.
Note that $\{v_1v_2,v_6v_2,v_4v_2\}\subseteq  EM(v_2)$,
$\{v_1v_3,v_5v_3,v_4v_3\}\subseteq  EM(v_3)$,
$\{v_6v_4\}\subseteq  EM(v_4)$,
$\{v_5v_1,v_6v_5\}\subseteq  EM(v_5)$, and hence $E(G^{\prime}_6)=\cup_{x\in M'_1 }EM(x)$,
it follows that
 $M'_1$ is a DEM set of $G^{\prime}_6$, and hence
$\operatorname{dem}(G^{\prime}_6) \leq |M'_1|=4$.
Let $M'$ be a DEM set of $G^{\prime}_6$ with  the  minimum  cardinality.
For the edge $v_{2i-1}v_{2i}\, (1\leq i\leq 3)$ and $w\in  \left(N(v_{2i-1})\cup  N(v_{2i})\right)\setminus\{v_{2i-1}v_{2i}\}$,
we have $d_{G- v_{2i-1}v_{2i}}(w,v_{2i-1})
=d_{G}(w,v_{2i-1})$
and $d_{G - v_{2i-1}v_{2i}}(w,v_{2i})
=d_{G}(w,v_{2i})$, and
so $v_{2i-1}v_{2i} \notin EM(w)$.
From Proposition \ref{Obs:CUV},
the edge
$v_{2i-1}v_{2i}\ (1\leq i\leq 3)$ is monitored by the vertex in $\{v_{2i-1}, v_{2i}\}$, and hence
$M' \cap \{v_{2i-1}, v_{2i}\}
\neq \emptyset$ ($1\leq i \leq 3$).
All sets $M' \in V(G^{\prime}_6)$ with $|M'|=3$ are shown in Table $1$.
Therefore, all sets $M'$ with $|M'|=3$
are not DEM sets of $G^{\prime}_6$, and hence $\operatorname{dem}(G^{\prime}_6)\geq 4$.
Therefore, we have $\operatorname{dem}(G^{\prime}_6)=4$.

\begin{table}[h]
\caption{The edges are not monitored by $M'$($|M'|=3$).}
\begin{center}
\tabcolsep 5pt
\begin{tabular}{|c|c|}
\hline $M'$ & $E(G'_6)- \cup_{x\in M'}EM(x)$ \\
\cline{1-2} $v_1,v_3,v_6$ & $v_2v_4$ \\
\cline{1-2} $v_1,v_4,v_5$ & $v_2v_6$ \\
\cline{1-2} $v_1,v_4,v_6$ & $v_3v_5$ \\
\cline{1-2} $v_2,v_3,v_5$ & $v_4v_6$ \\
\cline{1-2} $v_2,v_3,v_6$ & $v_1v_5$ \\
\cline{1-2} $v_2,v_4,v_5$ & $v_1v_3$ \\
\cline{1-2} $v_1,v_3,v_5$ &
$v_2v_6, v_2v_4, v_4v_6$ \\
\cline{1-2} $v_2,v_4,v_6$ & $v_1v_3, v_1v_5, v_3v_5$ \\
\cline{1-2}
\end{tabular}
\end{center}\vspace*{-3mm}
\end{table}

\medskip
For the graph $G^{\prime}_6  - v_3v_4$,
let $M_3=\{v_2,v_5\}$.
Note that $\{v_1v_2,v_6v_2,v_4v_2,v_1v_3\}\subseteq EM(v_2)$ and
$\{v_5v_1,v_6v_5,v_3v_5,v_6v_4\}
\subseteq  EM(v_5)$.
Since $E(G^{\prime}_6-v_3v_4)=\cup_{x\in M_3 }EM(x)$,
it follows that
 $M_3$ is a DEM set of $G^{\prime}_6$,
and hence $\operatorname{dem}(G^{\prime}_6 - v_3v_4)\leq 2$.
Since $G^{\prime}_6 - v_3v_4$ is not a tree,
it follows from Theorem \ref{th-dem-1} that $\operatorname{dem}(G^{\prime}_6 - v_3v_4)\geq  2$, and so $\operatorname{dem}(G^{\prime}_6 - v_3v_4)= 2$.
From Observation \ref{Obs:G_b},
$\operatorname{dem}(G_2)-
\operatorname{dem}(G_2 - v_3v_4)$
$=\operatorname{dem}(G^{\prime}_6 )-
\operatorname{dem}(G^{\prime}_6 - v_3v_4)=$
$4-2=2$, as desired.
\end{proof}

The \emph{friendship graph}, $Fr{(n)}$, can be constructed by joining $n$ copies of the complete graph $K_3$ with a common vertex, which is called the \emph{universal vertex} of $Fr(n)$.
Next, we give the proof of Theorem \ref{th-Ei}. \\

\noindent {\bf Proof of Theorem \ref{th-Ei}:}
Let $k,i$ be integers with $1\leq i \leq k$.
The graph $G^i$ is obtained by iteratively
adding an edge $u_iv_i$ to the graph
$G^{i-1}$.
Without loss of generality, let $G^0$ be the graph with
$V(G^{0})=\{c\}\cup \{u_j\,|\,1\leq j\leq k\}\cup \{v_j\,|\,1\leq j\leq k\}$ and $E(G^0)=\{cu_j,cv_j\,|\,1\leq j \leq k\}$,
and $G^i$ be the graph
with $V(G^{i})= V(G^{i-1})$ and $E(G^i)=E(G^{i-1})\cup \{u_iv_i\}$, where $1\leq i\leq k$.
Since $G^0$ is a tree, it follows from Theorem \ref{th-dem-1} that ${\rm dem}(G^0)=1$.
Note that the base graph of $G^1$ is
a complete graph $K_3$.
From Observation \ref{Obs:G_b}
and Theorem \ref{th-dem-n},
we have $\operatorname{dem}(G_1)=
\operatorname{dem}(K_3)=2$.

\medskip
Let $G=G^{i}$, where $2\leq i\leq k$.
Then $G_b=Fr(i)$.
Let
$M=\{u_t\,|\,1\leq t\leq i\}$.
From Theorem \ref{Th-Ncover}, we have $\{u_tv_t, cu_t\,|\,1\leq t\leq i\}\subseteq \cup_{x\in M}EM(x)$.
Since $2=d_{G}(u_1,v_t)\neq
d_{G-cv_t}(u_1,v_t)=3$ for $2\leq t\leq i$,
it follows that
$cv_t\in EM(u_1)$ for $2\leq t\leq i$.
Suppose that $t=1$.
Since $2=d_{G}(u_2,v_1)\neq
d_{G-cv_1}(u_2,v_1)=3$,
it follows that
$cv_1\in EM(u_2)$, and hence $E(G)\subseteq \cup_{x\in M}EM(x)$, and so $\operatorname{dem}(G)\leq i$.
Let $M$ be a DEM set of $G$ with  the  minimum  cardinality.
Note that $\left(N(u_j)\cup N(v_j)\right)\setminus\{u_j,v_j\}=\{c\}$.
Since $d_G(c,u_j)=d_{G-u_jv_j}(c,u_j)$
and $d_G(c,v_j)=d_{G-u_jv_j}(c,v_j)$
it follows that $u_jv_j \notin EM(c)$, where $1\leq j\leq k$.
From Proposition \ref{Obs:CUV},
the edge $u_jv_j$
is only monitored by $u_j$ or $v_j$, and hence $M\cap\{u_j,v_j\}\neq \emptyset$
for $1\leq j\leq k$,
Therefore, $\operatorname{dem}(G)\geq i$, and so $\operatorname{dem}(G)= i$.
Thus, there exists
a graph sequence  $\{G^i\,|\,0\leq i\leq k \}$,
with $e(G^i)-e(G^0)=i$ and $V(G^i)=V(G^j)$ for $0\leq i,j \leq k$,
such that
$\operatorname{dem}(G^{i+1})
-\operatorname{dem}(G^0)=i$, where $1\leq i\leq k-1$.\QED

\medskip
Foucaud et al.~\cite{FKKMR21} obtained the following result.
\begin{theorem}{\upshape \cite{FKKMR21}}
\label{ThmPnPM}
Let $\ell_1$ and $\ell_2$ be two integers with $\ell \geq 2$ and $\ell_2 \geq 2$. Then
$$
\operatorname{dem}\left(P_{\ell_1} \square P_{\ell_2}\right)=
\max \{\ell_1,\ell_2\}
$$
\end{theorem}

In the end of this section, we give the proof of Corollary \ref{cor-e}.

\noindent {\bf Proof of Corollary \ref{cor-e}:}
For any tree $T_n$,
$T_n+e_1$ is an unicyclic graph and $T_n+\{e_1,e_2\}$
is a tricyclic graph. From Theorems
\ref{th-dem-1} and \ref{Th-fes}, we have $\operatorname{dem}(T_n+e_1)
=\operatorname{dem}(T_n)+1=2$ and
$\operatorname{dem}(T_n+\{e_1,e_2\})=2$ or $3$.\QED

\section{The effect of deleted vertex}

A \emph{kipas} $\widehat{K}_n$ with $n\geq 3$ is the graph on $n+1$ vertices obtained from the join of $K_1$ and $P_n$, where $V(\widehat{K}_n)=\{v_0,v_1,\ldots,v_n\}$
and
$E(\widehat{K}_n)=\{v_0v_i\,|\,1\leq i\leq n\}
\cup \{v_{i}v_{i+1}\,|\,1\leq i\leq n-1\}$.

\begin{proposition}\label{Lemma:FAN GRAPH}
For $n\geq 7$, we have $\operatorname{dem}(\widehat{K}_n)
=\lfloor n/2 \rfloor$.
\end{proposition}

\begin{proof}
Let $P_n$ be the subgraph
of $\widehat{K}_n$
with vertex  set $\{v_{i} \,|\,1\leq i\leq n\}$ and edge set $\{v_{i}v_{i+1} \,|\,
1\leq i\leq n-1\}$.
First, we prove that $\operatorname{dem}(\widehat{K}_n)\geq
\lfloor n/2 \rfloor$.
Let $M$ be a DEM set of $\widehat{K}_n$ with  the  minimum  cardinality.
For any vertices $v_i,v_j\in V(\widehat{K}_n)$,
we have
\begin{equation*}\label{equ:dij}
d_{\widehat{K}_n}(v_i,v_j)=
\begin{cases}
1, & \mbox{if } i=0 \mbox{~or~} j=0\mbox{~or~} |i-j| =1; \\
2, & \mbox{if } 1\leq i, j \leq n \mbox{~and~} |i-j| \geq 2.
\end{cases}
\end{equation*}
For any edge $v_iv_{i+1}$ ($2\leq i\leq n-2$), we have
$\left(N_G(v_i)\cup N_G(v_{i+1})\right)\setminus
\{v_i,v_{i+1}\}=
\{v_{i-1},v_0,v_{i+2}\}$.
Since
$d_G(v_i,v_0)=d_{G-v_iv_{i+1}}(v_{i},v_0)=1$,
$d_G(v_{i+1},v_0)=d_{G-v_iv_{i+1}}(v_{i+1},v_0)=1$,
$d_{G-v_iv_{i+1}}(v_{i}, \linebreak v_{i-1})$
$=d_G(v_{i},v_{i-1})=1$,
$d_{G-v_iv_{i+1}}(v_{i+1},v_{i-1})$
$=d_G(v_{i+1},v_{i-1})=2$,
$d_{G-v_iv_{i+1}}(v_{i+1},v_{i+2})$
$=d_G(v_{i+1},v_{i+2})=1$,
and
$d_{G-v_iv_{i+1}}(v_{i},v_{i+2})$
$=d_G(v_{i},v_{i+2})=2$,
it follows that $v_iv_{i+1} \notin $
$ EM(v_{i+2})\cup EM(v_{i-1})\cup EM(v_{0})$.
From Proposition \ref{Obs:CUV},
the edge $\,v_iv_{i+1}\,$ can only be monitored
\eject

\noindent by the vertex in
$\{v_i,v_{i+1}\}$.
Similarly, the edge $v_iv_{i+1}$
is only monitored by the vertex in $\{v_i,v_{i+1}\}$, where $i=1, n-1$.
Therefore, $M\cap\{v_i,v_{i+1}\} \neq \emptyset$
for $1\leq i\leq n-1$, that is,
$M$ is a vertex cover set of
$P_n$.
Note that the vertex covering number
of $G$ is $\beta(G)$.
Since $\beta(P_n)=\lfloor n/2 \rfloor$,
it follows that
$\operatorname{dem}(\widehat{K}_n) \geq
\lfloor n/2 \rfloor$.

\medskip
Next, we prove that $\operatorname{dem}(\widehat{K}_n)
\leq \lfloor n/2 \rfloor$.
Let $M=\{v_i\,|\,i\equiv 0\pmod{2},1\leq i\leq n\}$.
For any edge $e \in E(P_n) \cup
\{v_0v_i\,|\, i\equiv 0\pmod{2},1\leq i\leq n\}$,
it follows from  Theorem \ref{Th-Ncover} that
$e$ is monitored by the vertex in $M$.
In addition,
for any edge $v_0v_i \in  \{v_0v_i\,|\,i\equiv 1\pmod{2},1\leq i\leq n\}$,
since $n\geq 7$, it follows that
there exists $j$ such that
$d_{G}(v_i,v_j)=2$
and $d_{G- v_0v_i}(v_{i},v_j)=3$,
where $j=i+3$ for $1\leq i\leq n-4$
and $j=2$ for $n-3\leq i\leq n$, and hence $v_0v_i \in EM(v_j)$.
Since any edge $v_0v_i\in E(\widehat{K}_n)$ can be monitored by the vertex in
$M$, it follows that $\operatorname{dem}(\widehat{K}_n) \leq \lfloor n/2 \rfloor$, and hence $\operatorname{dem}(\widehat{K}_n) =\lfloor n/2 \rfloor.$
\end{proof}

\noindent {\bf Proof of Theorem \ref{Obs:dv1}}
Note that
$\widehat{K}_{2k+2}=K_1
\vee P_{2k+2}$, where $V(K_1)=\{v_0\}$.
From Theorem
\ref{th-dem-1},
we have $\operatorname{dem}(P_{2k+2})=1$.
From Lemma \ref{Lemma:FAN GRAPH}, we have
$\operatorname{dem}(\widehat{K}_{2k+2})=k+1$, and hence
$\operatorname{dem}(\widehat{K}_{2k+2})
-\operatorname{dem}(\widehat{K}_{2k+2}-v_0)
=\operatorname{dem}(\widehat{K}_{2k+2})- \operatorname{dem}(P_{2k+2})=k.$
Let $G_1=\widehat{K}_{2k+2}$
and $H_1=P_{2k+2}$.
Then
$\operatorname{dem}(G_1)
-\operatorname{dem}(H_1)
=\operatorname{dem}(\widehat{K}_{2k+2})- \operatorname{dem}(P_{2k+2})=k$,
where $H_1$ is not a spanning subgraph
of $G_1$.

\medskip
Let $G_{2k+3}$ be a graph with vertex set
$V(G_{2k+3})=\{u_i\,|\,1\leq i\leq k+1\}
\cup \{v_i\,|\,0\leq i\leq k+1\}$
and edge set $E(G_{2k+3})=\{v_0u_i\,|\,1\leq i\leq k+1\}
\cup \{u_iv_i\,|\,1\leq i\leq k+1\}$.
Obviously, we have
$G_{2k+3}\setminus v_0 \cong (k+1)K_2$.
From  Observation
\ref{Obs:disjoint} and Theorem
\ref{th-dem-n}, we have $\operatorname{dem}(G_{2k+3}-v_0)
= \operatorname{dem}((k+1)K_2)= (k+1)\operatorname{dem}(K_2)=k+1$.
Since $G_{2k+3}$ is a tree,
it follows from Theorem \ref{th-dem-1} that
$\operatorname{dem}(G_{2k+3})=1$, and hence
$\operatorname{dem}(G_{2k+1}\setminus v_0)-
\operatorname{dem}(G_{2k+1})=k$.
Let $G_2=G_{2k+1}$
and $H_2=(k+1)K_2$.
Then
$\operatorname{dem}(H_2)
-\operatorname{dem}(G_2)
=\operatorname{dem}((k+1)K_2)- \operatorname{dem}(G_{2k+1})=k$,
where $H_2$ is not a spanning subgraph
of $G_2$, as desired. \smallskip \QED

\medskip
Note that $G_{2k+3}\setminus v_0 \cong (k+1)K_2$ is disconnected graph. For the connected graphs, we can also show that
there is a connected subgraph $H$
such that
$\operatorname{dem}(H)-\operatorname{dem}(G)$ can be arbitrarily large; see
Theorem \ref{th-G(k)}.

\medskip
The \emph{conical graph} $C(\ell,k)$ is a graph obtained by taking adjacency from
a center vertex $c$ to the first layer of
Cartesian product of $P_{\ell}$ and $C_k$, where $\ell \geq 1$ and $k\geq 3$.

\begin{figure}[!ht]
\centering
\includegraphics[width=5.7cm]{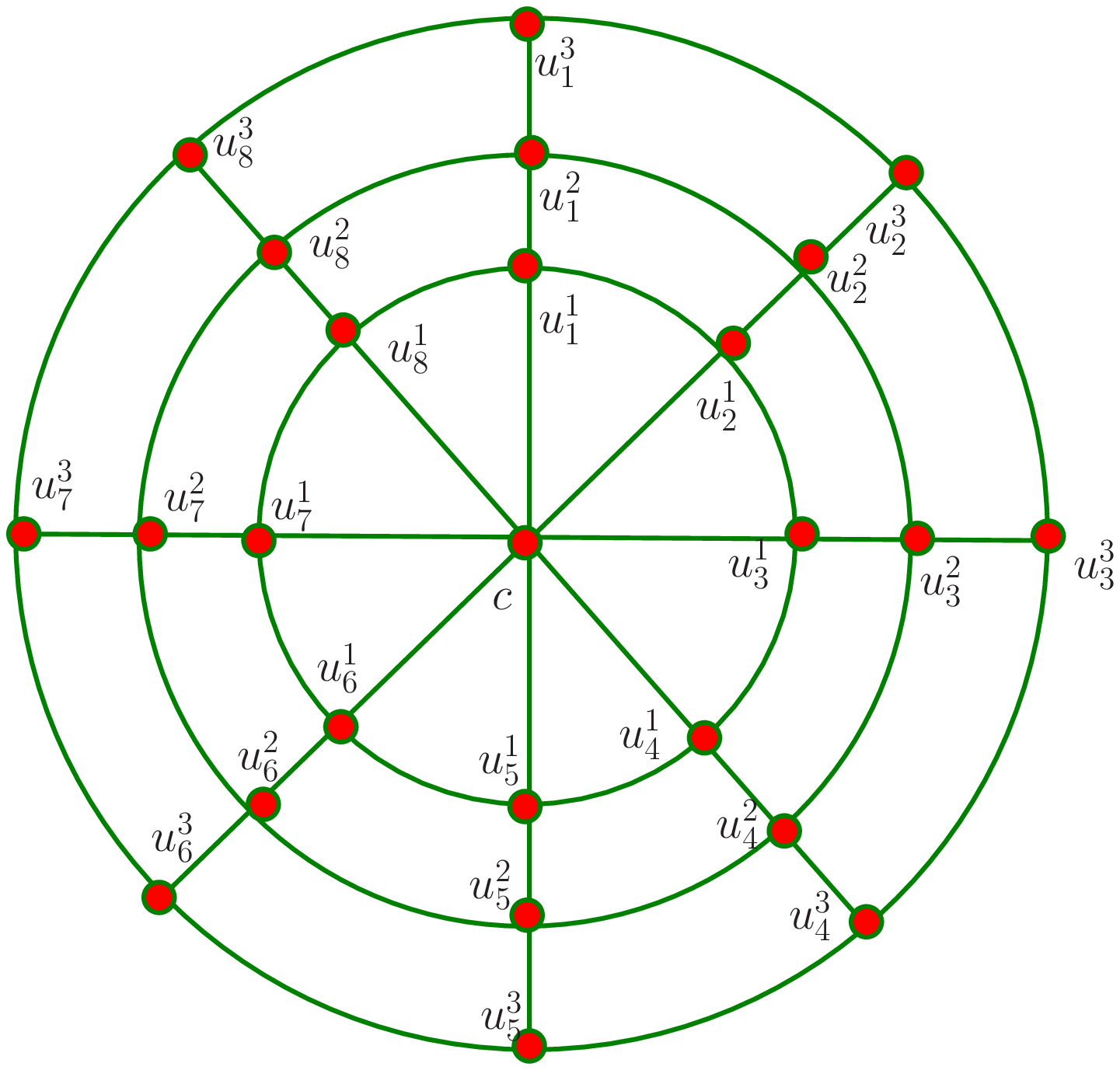}\vspace*{-.5cm}
\caption{The conical graph $C(3,8)$}
\label{Conical graph}\vspace*{-1mm}
\end{figure}

\medskip
Let the vertex set $V(C(\ell, k))=\{c\}\cup \{u_i^j\,|\,
1\leq i \leq k,1\leq j \leq \ell \}$
and the edge set $E(C(\ell, k))=(\cup_{i=1}^{\ell} E(C_i))\cup
(\cup_{i=1}^{k}{E(P_i)})$,
where $E(C_i)=\{ u_{k}^i u_1^i\}\cup
\{ u_j^i u_{j+1}^i\,|\,1\leq j\leq k-1\} $
($1\leq i\leq \ell$),
$E(P_i)=\{cu^1_i\}\cup
\{u^j_i u^{j+1}_i\,|\,1\leq j\leq \ell-1\}$
($1\leq i\leq k$).
The conical graph $C(3,8)$ is shown
in Figure \ref{Conical graph}.

\medskip
For $\ell=1$, the graph $C(1, k)$ is the wheel graph $W_k$, which is  formed
by connecting a single vertex $c$ to all the vertices of cycle $C_k$.
It is clear that $|V(C(\ell,k))|=k \ell+1$ and $e(C(\ell,k))=2 k \ell$.

\begin{lemma}
\label{Lem:EMCN}
Let $n\geq 3$ be an integer. For $v\in V(C_n)$,
we have
$$|EM(v)\cap E(C_n)|=
\begin{cases}
  n-1 & \mbox{if } n \text{ is odd}, \\
  n-2 &  \mbox{if } n \text{ is even}.
\end{cases}
$$
\end{lemma}
\begin{proof}
Let $G=C_n$ be cycle  with $V(G)=\{v_i\,|\, 1\leq i\leq n\}$
and  $E(G)=\{v_iv_{i+1}\,|\, 1\leq i\leq n-1\}\cup\{v_nv_1\}$.
Without loss of generality, let $v=v_1$.
Suppose that $n$ is odd and
$e_1=v_{\lfloor n/2\rfloor+1}v_{\lfloor n/2\rfloor+2}$.
Since $d_G(v_1,v_{\lfloor n/2\rfloor+1})=
d_{G-e_1}(v_1,v_{\lfloor n/2\rfloor+1})$
and $d_G(v_1,v_{\lfloor n/2\rfloor+2})=
d_{G-e_1}(v_1,v_{\lfloor n/2\rfloor+2})$,
it follows that $e_1\notin EM(v_1)$.
For any $e\in \{v_iv_{i+1}\,|\, 1\leq i\leq \lfloor n/2\rfloor\}$,
since $d_G(v_1,v_{i+1})\neq
d_{G-e}(v_1,v_{i+1})$,
it follows that  $e\in EM(v_1)$.
For any $e\in \{v_iv_{i+1}\,|\, \lfloor n/2\rfloor+2\leq i\leq n-1\}$,
since $d_G(v_1,v_{i})\neq
d_{G-e}(v_1,v_{i})$,
it follows that $e\in EM(v_1)$.
From Theorem
\ref{Th-Ncover}, we have $v_nv_1\in EM(v_1)$.
Therefore,
$EM(v_1)=\{v_1v_2, v_2v_3, \ldots$, $v_{\lfloor n/2\rfloor }
v_{\lfloor n/2\rfloor+1}$,
$ v_1v_n,v_nv_{n-1}, \ldots$, $v_{\lfloor n/2\rfloor+3}
v_{\lfloor n/2\rfloor+2}\}$,
and hence $|EM(v)\cap E(C_n)|=n-1$.
Similarly, if $n$ is even, then $|EM(v)\cap E(C_n)|=n-2$.
\end{proof}

\begin{theorem}\label{Thm:C*(n,k)}
For $k \geq 9$ and
$\ell \geq 2$,
we have
$$
\operatorname{dem}(C(\ell,k))=
\begin{cases}
 \sum_{i=1}^{\ell}{
\lceil k/(4i-2)\rceil}, & \mbox{if }
\ell \leq  a_k  ;\\[0.2cm]
\sum_{i=1}^{a_k}{
\lceil k/(4i-2)\rceil}+2(\ell-a_k), & \mbox{if }
\ell \geq a_k+1,
\end{cases}
$$
where $a_k=\lfloor k/4+(1+(-1)^{k+1})/8  \rfloor$.
\end{theorem}

\begin{proof}
Let $G=C(\ell, k)$ with
$V(G)=\{c\}\cup \{u_i^j\,|\,
1\leq i \leq k,1\leq j \leq \ell \}$
and $E(G)=(\cup_{i=1}^{\ell} E(C_i))\cup
(\cup_{i=1}^{k}{E(P_i)})$,
where $E(C_i)=\{ u_{k}^i u_1^i\}\cup
\{ u_j^i u_{j+1}^i\,|\,1\leq j\leq k-1\}$, $E(P_i)=\{cu^1_i\}\cup
\{u^j_i u^{j+1}_i\,|\,1\leq j\leq \ell-1\}$.
Let $M$ be a DEM set of $G$ with $M=\operatorname{dem}(C(\ell,k))$.
\begin{fact}\label{Lem:C_{n,k}}
For any vertex $v\in V(C_i)$,
we have
$$|EM(v)\cap E(C_i)|=
\begin{cases}
4i-2, & \mbox{if } 1 \leq i \leq
a_k;\\
k-2,  &
\mbox{if $k$ is even and } i \geq
 a_k+1; \\
k-1,   &
\mbox{if $k$ is odd and } i
\geq a_k+1.
\end{cases}
$$
and  $|EM(v)\cap E(C_j)|=0$, where $1\leq j\neq i\leq \ell$.
Furthermore, we have
$|EM(c)\cap E(C_i)|=0$.
\end{fact}
\begin{proof}
For any vertices $u^i_s, u^i_t\in V(C_i)$, where $1 \leq s, t\leq k$ and $1 \leq i  \leq \ell $, since there exists a path
$u^i_s\ldots,u^{1}_scu^1_t\ldots u^{i-1}_tu^{i}_t$ from $u^i_s$ to $u^i_t$, it follows that $d_{G}(u^i_s,u^i_t)\leq 2i$.
It is easy to see that there are two possible types of shortest paths $P_{u^i_su^i_t}$ from $u^i_s$ to $u^i_t$:

\medskip
\textbf{Type 1:} If $d_{C_i}(u^i_s,u^i_t)\geq 2i$,
then $P_{u^i_su^i_t}=u^i_s\ldots,u^{1}_scu^1_t\ldots u^{i-1}_tu^{i}_t$;

\medskip
\textbf{Type 2:} If $d_{C_i}(u^i_s,u^i_t)< 2i$,
then  the shortest path $P_{u^i_su^i_t}\sqsubseteq C_i$, where $C_i$ is subgraph of $G$.

\medskip\noindent
Therefore, $d_{G}(u^i_s,u^i_t)=\min
\{2i, d_{C_i}(u^i_s,u^i_t)\}$.
Suppose that
$1\leq i\leq a_k$.
For $s-2i+1\leq  t \leq s-1$, since $d_{C_i}(u^i_s,u^i_t)< 2i$, it follows that $d_{G}(u^i_s,u^i_t)=\min
\{2i, d_{C_i}(u^i_s,u^i_t)\}=d_{C_i}(u^i_s,u^i_t)$, and hence $d_{G}(u^i_s,u^i_tu^i_{t+1})= d_{C_i}(u^i_s,u^i_{t+1})$.
Thus, $d_{G}(u^i_s,u^i_tu^i_{t+1})=d_{G-u^i_{t}u^i_{t+1}}(u^i_s,u^i_tu^i_{t+1})$ and $d_{G-u^i_{t}u^i_{t+1}}(u^i_s,u^i_t)>d_{C_i}(u^i_s,u^i_{t+1})+1=d_{G}(u^i_s,u^i_t)$, and so
$u^i_{t}u^i_{t+1}\in EM(u^i_s)$.

\smallskip
Similarly, for any edge $u^i_{t}u^i_{t+1}\in E(C_i)$, where $s\leq  t \leq s+2i-2$,
since $d_G(u^i_s,u^i_{t+1}))\neq d_{G-u^i_{t}u^i_{t+1}}(u^i_s,u^i_{t+1})$,
it follows that $u^i_{t}u^i_{t+1}\in EM(u^i_s)$.
Therefore, $\{u^i_{t}u^i_{t+1}\,|\,s-2i+1\leq  t \leq s+2i-2\} \subseteq EM(u^i_{s})$,
where the subscripts are
taken modulo $k$, that is, $u^i_{k+1}=u^i_1$.\smallskip

\medskip\noindent
For any $u^i_ju^i_{j+1}\! \in\! E(C_i)-\left(\{u^i_{t}u^i_{t+1}\,|\,s-2i+1\leq t\leq s+2i-2\}\! \cup\! \{u^i_{s-2i}u^i_{s-2i-1}, u^i_{s+2i+1}u^i_{s+2i}\}\right)$,
since $d(u^i_s,u^i_{j})=d(u^i_s,u^i_{j+1})=2i$
and $d_{G - u^i_ju^i_{j+1}}(u^i_s,u^i_j)
=d_{G - u^i_ju^i_{j+1}}(u^i_s,u^i_{j+1})=2i$,
it follows that $u^i_ju^i_{j+1} \notin EM(u_i)$.
In addition,
let $e_1=u^i_{s-2i}u^i_{s-2i-1}$ and
$e_2=u^i_{s+2i-1}u^i_{s+2i}$.
Then,
$d_G(u^i_s,e_1)=d_G(u^i_s,u^i_{s-2i-1})=2i$
and $d_G(u^i_s,e_2)=d_G(u^i_s,u^i_{s+2i-1})=2i$.
Since
$d_{G- e_1}(u^i_s,u^i_{s-2i})$
$=d_{G}(u^i_s,u^i_{s-2i})=2i$
and
$d_{G- e_2}(u^i_s,u^i_{s+2i})$
$=d_{G}(u^i_s,u^i_{s+2i})=2i$,
it follows that $e_1,e_2 \notin EM(u_i)$, and hence $|EM(u^i_s)\cap E(C_i)|=|\{u^i_{t}u^i_{t+1}\,|\,s-2i+1\leq t\leq s+2i-2\}|=4i-2$
for any vertex $u^i_s\in V(C_i)$, where the subscripts are
taken modulo $k$.

\medskip
Suppose that
$i\geq a_k+1$.
%
For
$i\geq a_k+1$
and $1\leq s,t\leq k$,
if $u^i_s,u^i_t\in V(C_i)$,
then
$d_G(u^i_s,u^i_t)=d_{C_i}(u^i_s,u^i_t)$.
For any cycle $C_i$ with length $k$,
if $k$ is even,
then it follows from Lemma \ref{Lem:EMCN} that $|EM(v)\cap E(C_i)|=k-2$.
If $k$ is odd, then it follows from Lemma \ref{Lem:EMCN} that $|EM(v)\cap E(C_i)|=k-1$,
as desired.

For any vertex $u^i_s\in V(C_i)$ and any edge
$e=u^j_{m}u^j_{m+1} \in E(C_j)$, where $1\leq j \neq i\leq \ell$
and $1 \leq m\leq  k$,
since $d_{G- e}(u^i_s,u^j_{m})$
$=d_{G}(u^i_s,u^j_{m})$ and
$d_{G- e}(u^i_s,u^j_{m+1})$
$=d_{G}(u^i_s,u^j_{m+1})$,
it follows that $e \notin EM(u^i_s)$,
and hence $|EM(u^i_s)\cap E(C_j)|=0$.

\medskip
For any edge $e_3=u^i_su^i_{s+1} \in E(C_i)$,
we have $d_G(c,u^i_s)=d_{G}(c,u^i_{s+1})=i$,
and hence $d_{G- e_3}(c,u^i_s)
=d_{G- e_3}(c,u^i_{s+1})=i$, and so $e_3 \notin EM(c)$.
Therefore,
$|EM(c)\cap E(C_i)|=0$.
\end{proof}

Suppose that $\ell \geq a_k+1$. Since $e(C_i)=k$, it follows from Fact \ref{Lem:C_{n,k}} that $|M\cap E(C_i)|\geq 2$ for $a_k+1\leq i\leq \ell$ and
$|M\cap E(C_i)|\geq \lceil k/(4i-2)\rceil$ for $1\leq i\leq a_k$, and so
$\operatorname{dem}(G) \geq \sum_{i=1}^{a_k}{\lceil k/(4i-2)\rceil}+2(\ell-a_k)$.

Let $M=\cup_{i=1}^{i=k}{M_i}$,
where
$$M_i=
\begin{cases}
\{u^i_j\,|\,1\leq j \leq k,
j\equiv 1\bmod{(4i-2)}\}, & \mbox{if }  i \leq a_k;\\
\{u^i_1,u^i_{\lceil k/2\rceil}\}, &
\mbox{if } a_k +1 \leq i \leq \ell.  \\
\end{cases}
$$
Therefore, $e\in \cup_{x\in M_i}EM(x) \subseteq \cup_{x\in M}EM(x)$
for any edge $e\in E(C_i)$, where $1\leq i \leq \ell$. It suffices to prove that $e\in \cup_{x\in M}EM(x)$
for each edge in $E(P_i)$, where $1\leq i\leq k$.
For some vertex $u^1_i\in M_1$ and any $u^1_j\in V(G)$, where $1\leq i\neq j\leq k$,
if $j \in \{1,2, \ldots,i-3, i+3,\ldots,k\}$,
where the subscripts are
taken modulo $n$,
then
$d_G(u^1_i,u^1_j)\neq
d_{G-cu^1_j}(u^1_i,u^1_j)$, and hence $cu^1_j\in EM(u^1_i)$.
Similarly, for $1\leq t \leq \ell-1$,
since $d_G(u^1_i,u^{t}_j)\neq
d_{G-u^t_ju^{t-1}_j}(u^1_i,u^t_j)$ for $j\in \{1,2, \ldots,i-3, i+3,\ldots,k\}$,
it follows that $u^t_ju^{t-1}_j  \in EM(u^1_i)$, and hence $EM(u^1_i)
=\{u^1_iu^1_{i+1}, u^1_iu^1_{i-1}, cu_i^{1}\}\cup\{cu_{j}^1\,|\,j \in \{1,2, \ldots,i-3, i+3,\ldots,k\}\} \cup\{u^{p-1}_ju^{p}_{j}\,|\,2\leq p \leq \ell-1, j \in \{1,2, \ldots,i-3, i+3,\ldots,k\} \}$.
Without loss of generality,
let $u^1_1 \in M$. Then $ \cup_{i \in \{1,4,\ldots,k-2\}}E(P_i)\subseteq EM(u^1_1)$.
Since $EM(u^1_i)
=\{u^1_iu^1_{i+1}, u^1_iu^1_{i-1}, cu_i^{1}\}\cup\{cu_{j}^1\,|\,j \in \{1,2, \ldots,i-3, i+3,\ldots,k\}\}\cup \{u^{p-1}_ju^{p}_{j}\,|\,2\leq p \leq \ell-1, j \in \{1,2, \ldots,i-3, i+3,\ldots,k\} \}$ and $k\geq 9$, it follows that $E(P_2)\subseteq EM(u^1_{k-3})$, $E(P_3)\subseteq EM(u^1_{k-3})$, $E(P_2)\subseteq EM(u^1_{k-2})$, $E(P_3)\subseteq EM(u^1_{k-2})$, $E(P_k)\subseteq EM(u^1_{3})$ and $E(P_{k-1})\subseteq EM(u^1_{3})$

\medskip
For any edge $e \in E(P_i)$, where $i \in \{2,3,k,k-1\}$,
if $k$ is even,
then $e\in EM(u^1_3)\cup EM(u^1_{k-3})$;
if $k$ is odd, then $e\in EM(u^1_3)\cup EM(u^1_{k-2})$.
Therefore, $e\in \cup_{x\in M}EM(x)$
for any $e \in E(P_i)$, where $1\leq i \leq k$, and so
$\operatorname{dem}(G)\leq \sum_{i=1}^{a_k}{
\lceil k/(4i-2)\rceil}+2(\ell-a_k)$.
Thus, $\operatorname{dem}(G)=\sum_{i=1}^{a_k}{
\lceil k/(4i-2)\rceil}+2(\ell-a_k)$.
For $\ell \leq a_k$, it is
similar to the case that $\ell \geq a_k+1$,
as desired.
\end{proof}

\begin{theorem}\label{th-G(k)}
For any positive integer $k\geq 9$, there exists a connected graph $G$ such that
such that
$$
\operatorname{dem}(G\setminus v)-
\operatorname{dem}(G)=\lfloor k/2\rfloor-\lceil k/6\rceil,
$$
where $v\in V(G)$.
\end{theorem}

\begin{proof}
Let $G=C(2,k)$, where $\ell=2$ and
$k\geq 5$.
Note that $G\setminus v_0=C_k \square K_2$.
From Theorem \ref{ThmCnPn},
we have $\operatorname{dem}(G\setminus v_0)
=\operatorname{dem}(C_k \square K_2)=k$.
From
Theorem \ref{Thm:C*(n,k)},
we have
$\operatorname{dem}(C(2,k))=$
$\sum_{i=1}^{2}{\lceil k/(4i-2)\rceil}$
$=\lceil k/2\rceil+\lceil k/6\rceil $, and hence $\operatorname{dem}(G-v)-
\operatorname{dem}(G)
=k-\lceil k/2\rceil-\lceil k/6\rceil
=\lfloor k/2\rfloor-\lceil k/6\rceil$,
as desired.
\end{proof}

Let $G =C(\ell,k)$ and $H =C_{k}\Box P_{\ell}$.
From Theorems \ref{Thm:C*(n,k)}
and \ref{ThmCnPn},
if $\ell \gg k$, then
$\operatorname{dem}(G)/\operatorname{dem}(H)$ \mbox{$\approx 1$}.
From Theorems \ref{Thm:C*(n,k)}
and \ref{ThmCnPn}, if $k=402$ and $\ell=100$, then
$\operatorname{dem}(G)/
\operatorname{dem}(H)\approx 0.561453$.

\begin{corollary}\label{COR:RAT}
There exist two connected
graphs $H$ and $G$ such that
$$
\frac{\operatorname{dem}(G)} {\operatorname{dem}(H)}\approx 0.561453,
$$
where $H$ is an induced subgraph of $G$.
\end{corollary}

\noindent {\bf Proof of Theorem \ref{TH:deEV}:}
Let $G=\widehat{K}_{2k+2}$.
From Proposition \ref{Obs:dv1},
there exists a vertex $u \in V(G)$ such that
$\operatorname{dem}(G)
-\operatorname{dem}(G\setminus u)
=k$.
Note that $G\setminus u=P_{2k+2}$ is a connected graph.
In addition,
let $H=C(2,t)$, $k=\lfloor t/2\rfloor-\lceil t/6\rceil$ and $v\in V(H)$. From Theorem \ref{th-G(k)},
we have
$\operatorname{dem}(C(2,t)\setminus v)
-\operatorname{dem}(C(2,t))=k$, where $C(2,t)\setminus v=C_t\square K_2$
is a connected graph. \smallskip

In fact, $G\setminus v$ is a subgraph of $G$.
From Theorem \ref{TH:deEV},
for any positive integer $k\geq 3$,
there exists a graph $G$ such that
$\operatorname{dem}(G\setminus v)-
\operatorname{dem}(G)\geq k$.\smallskip \QED
\eject

Let $G=C(2,t)$ and $H_1=C(2,t)\setminus v$, where $t\geq 8$.
From Theorem \ref{TH:deEV},
$\operatorname{dem}(H_1)\geq \operatorname{dem}(G)$.
Note that $G$ is not tree,
it follows from Theorem \ref{th-dem-1} that $\operatorname{dem}(G)\geq 2$.
Let $H_2$ be a tree satisfying $H_2\sqsubseteq G$. From Theorem \ref{th-dem-1}, we have
$\operatorname{dem}(H_2)=1$, and hence $\operatorname{dem}(H_2) \leq \operatorname{dem}(G)$, and so
Corollary \ref{Cor:hsubg} holds.
\begin{corollary}\label{Cor:hsubg}
There exists a connected graph $G$ and two non-spanning subgraphs $H_1,H_2\sqsubseteq G$ such that
$\operatorname{dem}(H_1)\geq \operatorname{dem}(G)$ and $\operatorname{dem}(H_2) \leq
\operatorname{dem}(G)$.
\end{corollary}

\noindent {\bf Proof of Proposition \ref{pro-upper}:}
For any graph $G$ with order $n$
and $G\setminus v$ with at least one edge,
we have $\operatorname{dem}(G)\leq n-1$
and $\operatorname{dem}(G\setminus v)\geq 1$, and hence
$\operatorname{dem}(G)-
\operatorname{dem}(G\setminus v) \leq n-2$.
Furthermore, let $G=K_3$,
then $\operatorname{dem}(G)-
\operatorname{dem}(G\setminus v) = n-2$, and hence the upper bound is sharp.
Conversely, since $\operatorname{dem}(G)-
\operatorname{dem}(G\setminus v)=n-2$,
it follows that $\operatorname{dem}(G)=n-1$
and $\operatorname{dem}(G\setminus v)=1$.
From Theorem \ref{th-dem-1}, $G\setminus v$ is a tree.
Suppose that $|V(G)|\geq 4$. Since $\operatorname{dem}(G)=n-1$,
it follows from Theorem \ref{th-dem-1} that
$G=K_n$, and hence $G\setminus v=K_{n-1}$ which contradicts to the fact that $G\setminus v$ is a tree.
Suppose that $|V(G)|\leq 3$. Since $\operatorname{dem}(G)=n-1$, it follows that $G=K_n$, where $n \leq 3$.
If $G=K_2$, then $G\setminus v=K_1$, which contradicts to the fact that $G\setminus v$ contains at least
one edge.
Therefore, $G=K_3$, as desired. \medskip \QED

Next, we consider
the subgraph  $H$ of $G$.
If $H$ is a proper subgraph of $G$ satisfying
$\operatorname{dem}(H)\leq \operatorname{dem}(G)$,
then what is the relation between
$H$ and $G$?
A natural question is what is the maximum number of edges we can delete from
$G$ without changing the number of distance-edge monitoring?
We give a partial answer as follows.

Recall that the base graph $G_b$ is a subgraph of $G$ with $\operatorname{dem}(G) = \operatorname{dem}(G_b)$.
Therefore, we can give a lower bound for the edge set $E$ such that $\operatorname{dem}(G)=\operatorname{dem}(G-E)$.
\begin{observation}\label{Obs:DelEedeEQ}
Let $G$ be a connected graph, and let
$E_1=E(G)-E(G_b)$. For $E\subseteq E(G)$,
if $\operatorname{dem}(G)=\operatorname{dem}(G-E)$ and $G-E$ is a connected graph with
order at least $2$, then $|E| \geq |E_1|$.
\end{observation}

\noindent {\bf Proof of Theorem \ref{th-dem-2}:}
Let $E\subseteq E(G)$ satisfying $\operatorname{dem}(G)
=\operatorname{dem}(G-E)$ and $M=\{u,v\}$ be a DEM set of $G$ with $|M|=\operatorname{dem}(G)=2$. From Theorem \ref{Th-forest}, we have $|EM(u)|\leq n-1$ and $|EM(v)|\leq n-1$.
If $uv\in E(G)$, then $e(G) \leq 2(n-1)-1$.
Since $\operatorname{dem}(G)
=\operatorname{dem}(G-E)=2$, it follows from
Theorem  \ref{th-dem-1} that
$G-E$ must contain a cycle,
and hence $|E| \leq 2(n-1)-1-3=2n-6$.

Suppose that $uv\not\in E(G)$.
If $|EM(u) \cap EM(u)| \geq 1$,  then $|E| \leq 2n-6$, which is
similar to the case that $uv\in E(G)$.
If $EM(u) \cap EM(v)=\emptyset$, it follows from
Theorem \ref{Th-forest} that $e(G) \leq 2(n-1)$.
Since $\operatorname{dem}(G-E)=2$, then it follows from Theorem \ref{th-dem-1} that
$G-E$ must contain a cycle, and hence
$|E|\leq 2(n-1)-3=2n-5$.
Furthermore, we give the following claim.
\begin{claim}\label{claim1}
$|E|\leq 2n-6$.
\end{claim}
\begin{proof}
Assume, to the contrary, that
$|E|=2n-5$. Since $\operatorname{dem}(G-E)=2$, it follows from Theorem \ref{th-dem-1} that $G-E=C_3$.
Without loss of generality,
let $V(G-E)=\{v_1,v_2,v_3\}$.
In addition, from Theorem \ref{Th-forest}, the subgraph induced
by the edge set $EM(u)$ and $EM(v)$ are the
spanning trees of $G$.
If $u,v \in \{v_1,v_2,v_3\}$, then
$uv\in E(G)$, a contradiction.
Thus, $u \notin \{v_1,v_2,v_3\}$ or $v \notin \{v_1,v_2,v_3\}$.
Without loss  generality,
suppose that $u \notin \{v_1,v_2,v_3\}$.

If $d_{G}(u)=1$, then $|N(u)|=1$. Let $N(u)=\{w\}$. Since the subgraph induced
by the edge set $EM(u)$ and $EM(v)$ are
the spanning trees of $G$, it follows that $uw\in EM(u)\cap EM(v)$, which contradicts to the fact that $EM(u)\cap EM(v)=\emptyset$.
Therefore, $d_{G}(u)\geq 2$.
Since the subgraph induced
by the edge set $EM(v)$ is
a spanning tree of $G$,
it follows that there exists a vertex $u_1\in N(u)$
such that $uu_1\in EM(v)$.
From Theorem \ref{Th-Ncover}, we have $uu_1\in EM(u)$,
and hence $uu_1\in EM(u)\cap EM(v)$,
which contradicts to the fact that
$EM(u) \cap EM(u)=\emptyset$.
\end{proof}
From Claim \ref{claim1}, we have $|E|\leq 2n-6$. Furthermore, let $G=(n-2)K_1\vee K_2$ with vertex set
$V(G)=\{v_i\,|\,1\leq i \leq n\}$ and
edge set $E(G)=\{v_1v_2\}
\cup\{v_1v_i,v_2v_i|3\leq i\leq n\}$. Then, $\operatorname{dem}(G)=2$.
Let $E=\{v_1v_i,v_2v_i\,|\,4\leq i\leq n\}
\subseteq E(G)$.
From Observation
\ref{Obs:disjoint},
$\operatorname{dem}(G-E)=
\operatorname{dem}(K_3)$
$+(n-1)\operatorname{dem}(K_1)=2$,
and hence there exists an edge set $E_1$
such that $\operatorname{dem}(G-E)=2$
and $|E|=2n-6$, as desired.  \QED

\medskip
In the end of this section, we give the proof of Theorem \ref{The:sTN} as follows. \\[0.2cm]
\noindent {\bf Proof of Theorem \ref{The:sTN}:}
Let $G=K_n$ with vertex set
$\{v_i\,|\,1\leq i\leq n\}$ and
edge set $\{v_iv_j\,|\,1\leq i<j\leq n\}$.
Let $T$ be a spanning tree in $K_n$. For any edge $uv \in E(T)$ and
vertex $w\in \left(\left(
N_G(u)\cup N_G(v)\right)\setminus\{u,v\} \right)\cap V(T)$,
we have
$d_G(w,u)=d_G(w,v)=1$ and
$d_{G  - uv}(w,u)=d_{G - uv}(w,v)=1$, and hence $uv \notin EM(w)$.
From Proposition \ref{Obs:CUV},
any edge $uv\in E(T)$ is only monitored
by $u$ or $v$, and hence
$\operatorname{dem}(K_n|_T)\geq \beta{(T)}$.
From Theorem \ref{Theorem:Upperbond},
$\operatorname{dem}(K_n|_T)\leq \beta{(T)}$, and hence $\operatorname{dem}(K_n|_T)$
$=\beta{(T)}$.
Since $T$ is tree with
order $n$, it follows that $T$ is a bipartite graph.
Without loss of generality,
let $V(T)=U\cup V \ (|U|\leq |V|)$,
which is a bipartite partition of $V(T)$.
From the pigeonhole principle, we have
$|U|\leq \lfloor \frac{n}{2}\rfloor$.
For any $uv\in E(T)$,
we have $\{u,v\}\cap U\neq \emptyset$,
and hence $\beta{(T)}\leq \lfloor \frac{n}{2}\rfloor$.
In addition,
$T$ contains at least one edge, and hence $\beta{(T)}\geq 1$,
and so $1 \leq \beta{(T)} \leq \lfloor \frac{n}{2}\rfloor$.

\medskip
Suppose that $T=S_n$ with vertex set $\{v_i\,|\,1\leq i\leq n\}$ and edge set $E(S_n)=\{v_1v_i\,|\,2\leq i \leq n\}$.
Then $\{v_1\}$ is the vertex cover set of $S_n$, and hence $\beta{(T)}=1$,
and so the lower bound is sharp.
Suppose that $T=P_n$ with vertex set $\{v_i\,|\,1\leq i\leq n\}$ and edge set
$E(P_n)=\{v_iv_{i+1}\,|\,1\leq i \leq n-1\}$.
Then, $\{v_i\,|\,i\eqcirc 0\pmod 2, 1\leq i\leq n\}$ is
a minimum vertex cover set of $P_n$,
and hence $\beta{(T)}=\lfloor \frac{n}{2}\rfloor$,
and so the upper bound is sharp. \medskip \QED

Similar to Theorem \ref{The:sTN}, we have the following corollary.
\begin{corollary}\label{Col:CSH}
Let $H$  be a subgraph of $G$ with $|V(H)|=p$. Then,
$$1 \leq \operatorname{dem}(G|_H) \leq
p-1.$$
Furthermore, the bounds are sharp.
\end{corollary}

\begin{theorem}
If $H$ is a connected induced subgraph
of graph $G$, then
$$
\operatorname{dem}(G)-\operatorname{dem}(G|_H) \leq |V(G)|-|V(H)|.
$$
Moreover, if $G$ and $H$ are both complete graphs,
then the bound is sharp.
\end{theorem}
\begin{proof}
For any graph $G$ and $H$, where
$H$ is an induced subgraph of $G$.
Let $M_1 \subseteq V(H)$ be a restrict-DEM set of $H$ in $G$ with $|M_1|=\operatorname{dem}(G|_H)$.
Let $M=(V(G)-V(H))\cup M_1$.
We will prove that $\operatorname{dem}(G)\leq |M|$.
For any edge $uv\in E(G)$,
if $u$ or $v$ in $V(G)-V(H)$, then
it follows from Theorem \ref{Th-Ncover} that
$e$ is monitored by the vertex in $V(G)\setminus V(H)$.
For any edge $e\in E(H)$,
since $M_1$ is a restrict-DEM set, it follows that $e$ is monitored by
the vertex in $M_1$, and hence $M$ is a DEM set in $G$.
Since $|M|=|M_1|+(|V(G)|-|V(H)|)=|V(G)|-|V(H)|
+\operatorname{dem}(G|_H)$,
it follows that $\operatorname{dem}(G)\leq |M|=\operatorname{dem}(G|_H)
+(|V(G)|-|V(H)|)$, and
so $\operatorname{dem}(G)-
\operatorname{dem}(G|_H)\leq
(|V(G)|-|V(H)|)$.
Furthermore, let $G=K_n$ and $H=K_m$ ($3\leq m\leq n$).
From Theorem \ref{th-dem-n},
$\operatorname{dem}(G)=n-1$.
For any $uv\in E(H)$ and $w \in
\left(N(u)\cup N(v)\right)\setminus\{u,v\}$,
we have
$d_G(w,u)=d_G(w,v)=1$ and
$d_{G - uv}(w,u)=d_{G - uv}(w,v)=1$,
and hence $uv \notin EM(w)$.
From Proposition \ref{Obs:CUV},
any edge $uv\in E(H)$ is only monitored
by $u$ or $v$, and so
$\operatorname{dem}(K_n|_{K_m})\geq \beta{(K_m)}=m-1$.
From Theorem \ref{Theorem:Upperbond},
$\operatorname{dem}(K_n|_{K_m})\leq \beta{(K_m)}=m-1$, and hence $\operatorname{dem}(K_n|_{K_m})=m-1$.
Therefore, $\operatorname{dem}(G)-\operatorname{dem}(G|_H)= (n-1)-(m-1)=|V(G)|-|V(H)|$, as desired.
\end{proof}

\section{Perturbation results for some known graphs}

Firstly, we study the change of DEM numbers for some well-known graphs when any edge (or vertex) of the graph is deleted.

\subsection{Deleting one edge or vertex from some known graphs}
Let $NV_i(G)=\{v\,|\,d_{G}(v)=i,v\in V(G)\}$ and $NE_{a,b}(G)=\{uv\,|\,uv\in E(G), d_{G}(u)=a,d_{G}(v)=b\}$.
Note that
$E(P_n)=NE_{1,2}(P_n)\cup NE_{2,2}(P_n)$.
If $e \in NE_{1,2}(P_n)$, then
$\operatorname{dem}\left(P_n- e\right)=1$,
and hence
$\operatorname{dem}\left(P_n\right)
-\operatorname{dem}\left(P_n- e\right)
=0.$
If $e  \in NE_{2,2}(P_n)$, then
$\operatorname{dem}\left(P_n -  e \right)=2$,
and hence
$\operatorname{dem}\left(P_n\right)
-\operatorname{dem}\left(P_n -  e \right)
=-1$.
\begin{corollary}\label{cPn}
Let $P_n$ be a path of order $n$, where $n\geq 2$.
For any $e\in E(P_n)$,
we have
$$
\operatorname{dem}(P_n- e)=
\begin{cases}
\operatorname{dem}(P_n), &
 \mbox{if $e  \in NE_{1,2}(P_n)$;}  \\
\operatorname{dem}(P_n)+1,
& \mbox{if $e \in NE_{2,2}(P_n)$}.
\end{cases}
$$
\end{corollary}

Foucaud et al.~\cite{FKKMR21} obtained the
DEM numbers of
complete bipartite graph $K_{\ell_1,\ell_2}$.
\begin{theorem}{\upshape \cite{FKKMR21}}
\label{ThmKMN}
Let $\ell_1$ and $\ell_2$ be two integers with $\ell \geq 1$ and $\ell_2 \geq 1$. Then
$$
\operatorname{dem}\left(K_{\ell_1,\ell_2}\right)
=\min \{\ell_1,\ell_2\}.
$$
\end{theorem}

The following corollary is immediate.
\begin{corollary}
Let $n\geq 3$ be an integer. Then,

$(i)$ for any edge $e\in E(C_n)$,
$\operatorname{dem}(C_n- e)=\operatorname{dem}(C_n)-1=1$;

$(ii)$ for any edge $e\in E(K_n)$,
$\operatorname{dem}(K_n- e)=\operatorname{dem}(K_n)-1=n-2$;

$(iii)$ for any edge $e\in E(K_{n,n})$,
$\operatorname{dem}(K_{n,n}- e)=\operatorname{dem}(K_{n,n})=n$.
\end{corollary}

\begin{proof}

\vspace*{-4mm}
$(i)$ From Theorem \ref{th-dem-1}, we have
$\operatorname{dem}(C_n- e)=
\operatorname{dem}(P_n)= 1$. Since $\operatorname{dem}(C_n)=2$, it follows that
$\operatorname{dem}(C_n- e)
=\operatorname{dem}(C_n)-1=1$.

\medskip
$(ii)$ Let $G=K_n$.
From Theorem \ref{th-dem-n},
we have $\operatorname{dem}(G)=n-1$ and $\operatorname{dem}(G-uv)\leq n-2$ for any edge $uv\in E(G)$.
Then, we prove that $\operatorname{dem}(G-uv)\geq n-2$.
Suppose that $\operatorname{dem}(G-uv)\leq n-3$.
Let $M$ be a DEM set of $G-uv$ with the  minimum  cardinality.
For any edge $xy \in E(G - uv)$ and
vertex $w\in (N_G(x)\cup N_G(y))\setminus\{x,y\}$,
we have
$d_G(w,x)=d_{G - xy}(w,x)$ and
$d_G(w,y)=d_{G - xy}(w,y)$,
and hence $xy \notin EM(w)$.
From Proposition \ref{Obs:CUV},
any edge $xy\in E(G-uv)$
is only monitored by $x$ and $y$,
and hence $M\cap\{x,y\}\neq \emptyset$.
If $u,v\in M$, then there exist two vertices $v_1,v_2\in V(G)\setminus \{u,v\}$
such that $v_1,v_2\notin M$, and hence $v_1v_2\notin \cup_{x\in M}EM(x)$, which contradicts to the fact that $\cup_{x\in M}EM(x)=E(G-uv)$.
Suppose that $u$ or $v\notin M$. Without loss of generality, let $u\notin M$. Then there exists a vertex $v_1\in V(G)\setminus \{u,v\}$
such that $v_1\notin M$, and hence $uv_1\notin \cup_{x\in M}EM(x)$, which contradicts to the fact that $\cup_{x\in M}EM(x)=E(G-uv)$.
Therefore, $\operatorname{dem}(G-uv)\leq n-2$, and so $\operatorname{dem}(G-uv)= n-2$.
This implies that
$\operatorname{dem}(K_n- e)=\operatorname{dem}(K_n)-1=n-2$.

\medskip
$(iii)$ Let $G=K_{n,n}$ with
vertex set $V(G)=\{u_i\,|\,1\leq i\leq n\}$
$\cup \{v_i\,|\,1\leq i\leq n\}$
and edge set $E(G)=\{u_iv_j\,|\,1\leq i,j \leq n\}$.
From Theorem
\ref{ThmKMN}, we have $\operatorname{dem}(K_{n,n})=n$.
Without loss of generality,
let $G_1=G-u_1v_1$. Firstly, we prove that $\operatorname{dem}(G_1)\geq n$.
Suppose that $M$ is a DEM set of $G_1$ with $|M|=n-1$.
Then, there exists two vertices $u_p, v_q\notin M$ with $u_pv_q\in E(G_1)$.
For any $u\in M\cap \{u_i\,|\,1\leq i\leq n\}$, we have $d_{G_1-u_pv_q}(u,u_p)=d_{G_1}(u,u_p)=2$ and $d_{G_1-u_pv_q}(u,v_q)=d_{G_1}(u,v_q)=1$, and hence $u_pv_q\notin EM(u)$. Similarly, $u_pv_q\notin EM(v)$ for any
$v\in M\cap \{v_i\,|\,1\leq i\leq n\}$, and hence $u_pv_q\notin \cup_{x\in M}EM(x)$, which contradicts to the fact that $\cup_{x\in M}EM(x)=E(G_1)$. Therefore, $\operatorname{dem}(G_1)\geq n$.
Then, we prove that $\operatorname{dem}(G_1)\leq n$.
Let $M=\{u_i\,|\,1\leq i\leq n\}$.
Then $M$ is a vertex cover set
of $G_1$.
From Theorem \ref{Theorem:Upperbond},
$\operatorname{dem}(G_1)\leq |M|=n$,
and hence $\operatorname{dem}(G_1)=n$.
Therefore, $\operatorname{dem}(K_{n,n}- e)=\operatorname{dem}(K_{n,n})=n$.
\end{proof}

There exist three graphs $G_1,G_2,G_3$
and $v\!\in V(G)$ such that
$\operatorname{dem}(G_1)\!>\! \operatorname{dem}(G_1-v)$,
$\operatorname{dem}(G_2)\! = \operatorname{dem}(G_2-v)$,
$\operatorname{dem}(G_3) < \operatorname{dem}(G_3-v)$, respectively.

The following corollary is immediate.
\begin{corollary}\label{corL1}
Let $n \geq 3$ be an integer.
For any $v\in v(G)$, we have\medskip

$(i)$ $\operatorname{dem}(C_n\setminus v)=
\operatorname{dem}(C_n)-1=1$;

$(ii)$ $\operatorname{dem}(K_n\setminus v)
=\operatorname{dem}(K_n)-1=n-2$;

$(iii)$ $\operatorname{dem}(K_{n,n}\setminus v)
=\operatorname{dem}(K_{n,n})-1=n-1$.
\end{corollary}

\begin{proposition}\label{cPnv1}
Let $P_n$ be a path with vertex set $\{v_i\,|\, 1\leq i\leq n\}$ and edge set $\{v_iv_{i+1}\,|\,1\leq i\leq n-1\}$, where $n\geq 5$.
For any $v\in V(P_n)$,
we have
$$\operatorname{dem}(P_n\setminus v)=
\begin{cases}
  \operatorname{dem}(P_n), &
\mbox{if $v \in \{v_1, v_2,v_{n-1},v_n\}$};  \\
\operatorname{dem}(P_n)+1, &
\mbox{if $v \in \{v_i\,|\,3\leq i\leq n-2\}$.}
\end{cases}
$$
\end{proposition}
\eject
\begin{proof}
From Theorem \ref{th-dem-1}, we have $\operatorname{dem}\left(P_n\right)=1$.
If $v \in \{v_1,v_n\}$,
then
$\operatorname{dem}\left(P_n\setminus v\right)=1$.
If $v \in \{v_2,v_{n-1}\}$,
then
$\operatorname{dem}\left(P_n\setminus v\right)
=\operatorname{dem}\left(P_{n-1}\right)+
\operatorname{dem}\left(K_1\right)=1$,
where $V(K_1)=\{v_1\}$ or $V(K_1)=\{v_n\}$, and hence $\operatorname{dem}\left(P_n\right)
-\operatorname{dem}\left(P_n\setminus v\right)=0$.
If $v \in  \{v_i\,|\,3\leq i\leq n-2\}$, then
$\operatorname{dem}\left(P_n- v\right) =2$,
and hence
$\operatorname{dem}\left(P_n\right)
-\operatorname{dem}\left(P_n\setminus v\right)=-1$.
\end{proof}

Proposition \ref{cPnv1} shows
that there exists a graph $G$ and $v\in V(G)$
such that $\operatorname{dem}(G) = \operatorname{dem}(G\setminus v)$ or
$\operatorname{dem}(G) < \operatorname{dem}(G\setminus v)$.

\subsection{Whether the DEM set is still applicable?}

Foucaud et al.~proved in \cite{FKKMR21} that the problem
DEM SET is $NP$-complete.
For any graph $G$, a natural question is whether the original DEM set  can monitor
all edges if some edges or vertices in $G$ are delated. We design the Algorithm \ref{algorithm:EM_veq} and
the time complexity  is polynomial.
\begin{framed}
\noindent
WHETHER THE DEM SET IS STILL APPLICABLE?\\
Instance: A graph $G=(V, E)$, an edge
$e\in E(G)$  and a DEM
set $M$ of $G$.\\
Question:  Whether $M$ is still
a DEM set
for the graph $G-e$?
\end{framed}
Given a graph $G$, a DEM set $M$,
and an edge $e\in E(G)$,
our goal is to determine whether the
original DEM set $M$
is still valid in the resulting graph $G- e$.
The algorithm is shown
in Algorithm~\ref{algorithm:EM_veq}.

\begin{algorithm}
\caption{The algorithm for determining
$M$ is or not
monitor set for $G- e$}
\label{algorithm:EM_veq}
 \begin{algorithmic}[1]
\Require A graph $G$, $M\subseteq V(G)$
and $e\in E(G)$;
\Ensure $E(G- e) \subseteq \cup_{x\in M}{EM(x)}$ is TRUE or FALSE;
\State $M_1$ $\gets E(G- e)$
\For{each vertex $v\in M$}
\State $M_1$ $ \gets$ $M_1-EM(v)$
\EndFor

\If{$M_1$  $=$ $\emptyset$}
\Return  $E(G- e) \subseteq \cup_{x\in M}{EM(x)}$ is TRUE;
\Else{}
\Return
$E(G- e) \subseteq \cup_{x\in M}{EM(x)}$ is FALSE;
\EndIf
\end{algorithmic}
\end{algorithm}

The algorithm of how to compute the edge set $EM(x)$ from $G$ is polynomial by
the breadth-first spanning tree algorithm.
Hence the time complexity of Algorithm \ref{algorithm:EM_veq} is polynomial.

\section{Conclusion}

In this paper, we studied
the effect of deleting edges and vertices in a graph $G$ on the DEM number.
We obtained that
$\operatorname{dem}(G-e)-\operatorname{dem}(G) \leq 2$ for any graph $G$ and $e \in E(G)$.
Furthermore, the bound is sharp. In addition, we can find a graph $H$ and
$v\in V(H)$ such that $\operatorname{dem}(H\setminus v)-\operatorname{dem}(H)$ can be arbitrarily large.
This fact gives an answer to the monotonicity of the DEM number.
This means that there exist two graphs $H$ and $G$ with $H\sqsubseteq G$ such that $\operatorname{dem}(H) \geq \operatorname{dem}(G)$.

\medskip
It is interesting to consider the following problems for future work.
\begin{itemize}
\item[] $(1)$ Characterize the graphs $\operatorname{dem}(H)\geq \operatorname{dem}(G)$ if $H \sqsubseteq G$.

\item[] $(2)$ For a graph $G$ and $E\subseteq E(G)$, what is
the maximum value of $|E|$ such
that $\operatorname{dem}(G)= \operatorname{dem}(G-E)$?

\item[] $(3)$ For any $\epsilon>0$, whether the ratio $\frac{\operatorname{dem}(G) } {\operatorname{dem}(H)}\leq \epsilon$ holds, where $H$ is an induced subgraph of $G$.
\end{itemize}
In addition, it would be interesting to study distance-edge monitoring
sets in further standard graph classes, including circulant graphs, graph products,
or line graphs. In addition, characterizing the graphs with $\operatorname{dem}(G)=n-2$
would be of interest, as well as clarifying further the relation of the
parameter $\operatorname{dem}(G)$ to other standard graph parameters, such as
arboricity, vertex cover number and feedback edge set number.\\[0.2mm]

\noindent \textbf{Acknowledgements.} We would like to thank the anonymous referees for suggesting the problem and for several helpful comments regarding this paper.

\end{document}